\documentclass{article}
\usepackage{fullpage,epsfig}
\usepackage{color}
\usepackage{bookmark}
\usepackage{subfigure}
\usepackage{amssymb,amsthm,amsmath}
\usepackage{bbm}

\usepackage[round,sort]{natbib}

\usepackage{hyperref}
\hypersetup{
    colorlinks=true,
    linkcolor=blue,
    citecolor=blue,
    urlcolor=blue
}
\renewcommand{\eqref}[1]{(\ref{#1})}

\newcommand{\C}{{\mathbb C}}

\newcommand{\E}{{\mathbb E}}
\newcommand{\bbS}{{\mathbb S}}
\newcommand{\bbF}{{\mathbb F}}

\newcommand{\calN}{{\mathcal N}}

\newcommand{\calS}{{\mathcal S}}

\newcommand{\Prob}{{\rm Prob}}

\newcommand{\tr}{{\rm tr}}
\newcommand{\re}{{\rm Re}}
\newcommand{\im}{{\rm Im}}
\newcommand{\dist}{{\rm dist}}

\newcommand{\vx}{{\mathbf x}}

\newcommand{\vz}{{\mathbf z}}

\newcommand{\vb}{{\mathbf b}}
\newcommand{\vv}{{\mathbf v}}
\newcommand{\vu}{{\mathbf u}}

\newcommand{\vh}{{\mathbf h}}
\newcommand{\vs}{{\mathbf s}}
\newcommand{\vg}{{\mathbf g}}
\newcommand{\ve}{{\mathbf e}}
\newcommand{\0}{{\mathbf 0}}

\title{The Fusion Frame Phase Retrieval}
\author{
Haixia Liu\thanks{School of Mathematics and Statistics  \& Institute of Interdisciplinary Research for Mathematics and Applied Science \& Hubei Key Laboratory of Engineering Modeling and Scientific Computing, Huazhong University of Science and Technology, Wuhan, Hubei, China. Email: liuhaixia@hust.edu.cn. The work of H.X. Liu was supported in part by Interdisciplinary Research Program of HUST 2024JCYJ005, National Key Research and Development Program of China 2023YFC3804500.},
Bing Gao\thanks{School of Mathematical Sciences, Nankai University, Tianjin 300071, China. Email: gaobing@naikai.edu.cn. The work of B. Gao was supported in part by National Natural Science Foundation of China No. 12001297.},
Yang Wang\thanks{Department of Mathematics, The University of Hong Kong, Pokfulam Road, Hong Kong. Email:yang.wang@hku.hk.}}
\usepackage{amsthm}
\usepackage{multirow}
\usepackage{natbib}
\usepackage{algorithm,algpseudocode}
\newtheorem{theorem}{Theorem}[section]

\newtheorem{lemma}[theorem]{Lemma}
\newtheorem{definition}[theorem]{Definition}

\numberwithin{equation}{section}
\allowdisplaybreaks

\begin{document}
\maketitle
\begin{abstract}
The phase retrieval problem involves reconstructing a function or signal solely from the magnitude of linear measurements. Most theoretical analyses of phase retrieval algorithms rely on i.i.d. Gaussian random measurements or sub-Gaussian random measurements. In this paper, our focus is on the fusion frame phase retrieval problem, where the sampling matrices are i.i.d. rank-$r$ orthogonal projections drawn from the Haar measure. We present concentration inequalities for functions on the set of rank-$r$ orthogonal projection matrices. These inequalities are crucial for the theoretical analysis of the fusion frame phase retrieval problem. Based on these inequalities, we demonstrate that gradient descent, combined with a two-stage initialization, achieves linear convergence to the target signal up to a global phase with a measurement complexity of $O(d\log^2 d)$ when the rank $r = O(1)$. We verify this convergence through numerical results.
\end{abstract}
% \keywords{Phase retrieval, fusion frame, orthogonal projection matrix, convergence}
% \maketitle
	\section{Introduction}
	\setcounter{equation}{0}
	The classic phase retrieval problem involves the reconstruction of a function from the magnitude of its Fourier transform. Let $f(\vx) \in L^2(\mathbb{R}^d)$, and it is well-known that $f$ can be uniquely reconstructed from its Fourier transform, denoted as $\mathcal{F}f$. However, in certain applications such as $X$-ray crystallography, only the magnitude $|\mathcal{F}f|$ of the Fourier transform can be measured, while the phase information is lost. This raises the question of whether it is possible to reconstruct $f$, thereby recovering the lost phase information, up to some inherent ambiguities such as translation and reflection.
	
	This finite-dimensional formulation of the phase retrieval problem has gained significant attention in recent years. In our paper, we address the phase retrieval problem in this setting. Specifically, given a frame $\mathcal{F} = \{\mathbf{f}_1, \mathbf{f}_2, \ldots, \mathbf{f}_N\}$ in $\bbF^d$, where $\bbF$ can be either $\mathbb{R}$ or $\mathbb{C}$, the goal is to reconstruct $\mathbf{x}$ from the magnitude of the inner products, denoted as $|\langle \mathbf{x}, \mathbf{f}_j \rangle|,j=1,\ldots,N$. 
	
	A particular class of phase retrieval problems that has received extensive research attention is when the measurements are i.i.d. Gaussian random measurements. For such problems, efficient reconstructive algorithms based on convex relaxation techniques have been developed. Some notable methods in this context include PhaseLift \citep{candes2014solving,candes2015phase1}, PhaseCut and MaxCut \citep{ waldspurger2015phase}, PhaseMax \citep{goldstein2018phasemax}, and the work by Bahmani and Romberg \citeyearpar{bahmani2017phase}.

	These convex methods have demonstrated promising performance in solving phase retrieval problems with i.i.d. Gaussian measurements. However, they can become computationally challenging for large-dimensional problems or high computational complexity. As a result, various non-convex optimization approaches have been developed to address these challenges and improve efficiency.
	
	Among the non-convex methods, AltMinPhase \citep{netrapalli2015phase} and Kaczmarz \citep{wei2015solving} estimate the missing phase information. AltMinPhase solves the phase retrieval problem through the least squares method, while Kaczmarz employs the Kaczmarz method. It has been shown that a resampled version of AltMinPhase requires $O(d\log^3 d)$ i.i.d. random Gaussian measurements to geometrically converge to the true solution up to a unimodular scalar. The Wirtinger Flow (WF) algorithm introduced by Cand\`es et al. \citeyearpar{candes2015phase} is another non-convex method that is guaranteed to converge linearly to the global minimizer for Gaussian measurements when the number of measurements $N$ is on the order of $O(d\log d)$.
	
	Various other techniques, such as truncated methods \citep{chen2015solving, wang2017solving}, have been developed to improve efficiency and robustness, particularly when the number of measurements $N$ is on the order of $O(d)$ for Gaussian measurements. Additionally, Riemannian Optimization \citep{wei2016Guarantees,cai2024solving}, Gauss-Newton's method \citep{gao2017phaseless}, rank-$1$ alternating minimization algorithm \citep{cai2019fast}, and composite optimization algorithm \citep{duchi2019solving} have provided theoretical convergence analysis for Gaussian random measurements.
	
	Some of the aforementioned methods, such as the WF algorithm, also apply to Fourier measurements with a specially designed random mask known as the Coded Diffraction model \citep{candes2015phase}. However, it is worth noting that these models represent a small fraction of the phase retrieval problem space, and comparatively fewer guarantees are available for structured non-Gaussian measurement ensembles. In order to address this gap, recent research has focused on analyzing the phase retrieval problem with sub-Gaussian measurements \citep{gao2021phase,krahmer2020complex, li2021phase}. 
	
	The phase retrieval problem is a specific instance of the broader task of recovering a vector $\mathbf{x} \in \bbF^d$ from quadratic measurements of the form $\{\mathbf{x}^* A_j \mathbf{x}\}_{j=1}^N$, where each $A_j$ is a Hermitian matrix in $\bbF^{d \times d}$.  This problem, known as the generalized phase retrieval problem, was initially explored from a theoretical perspective by Wang and Xu \citeyearpar{wang2019generalized}. Prior to that, specific cases of this problem had been investigated, such as the one involving orthogonal projection matrices $\{A_j\}_{j=1}^N$, which was studied by other researchers \citep{edidin2017projections, heinosaari2013quantum, cahill2016phase}. This paper adopts an algorithmic viewpoint on the fusion frame phase retrieval problem, which we formally define as follows:
	
	\medskip
	\noindent
	{\bf The Fusion Frame Phase Retrieval Problem.}~~{\em Let $A_1,A_2,\ldots,A_N$ be i.i.d. rank-$r$ orthogonal projections in $\C^d$ drawn from the Haar measure. Can we reconstruct any $\vx\in\C^d$ up to a unimodular scalar from $\{\vx^*A_j\vx\}_{j=1}^N$, and if so, how?}
	
	\medskip
	
	Our aim is to demonstrate that the gradient descent method, coupled with two-stage initialization in solving fusion frame phase retrieval, achieves linear convergence to the global minimizer. To sum up, our contributions in this work include:
	
	\noindent
	\textbf{We give the concentration inequalities of functions on the set of rank-$r$ orthogonal projection matrices.} To establish the convergence of functions on the set of rank-$r$ orthogonal projection matrices, concentration inequalities play a crucial role. These concentration inequalities are provided in Theorems \ref{theo:concentration1} and \ref{theo:concentration1-1}, as well as Lemmas \ref{lem:concentration1} and \ref{lem:concentration1-1}.
	
	\noindent
	\textbf{We provide the proof of linear convergence for the fusion frame phase retrieval problem.} 
	To ensure linear convergence, it is crucial to initialize the algorithm near the global minimizer. To this end, we consider two-stage initialization. The gradient descent achieves linear convergence with a measurement complexity of $O(d\log^2 d)$ when the rank $r = O(1)$.
	
	The remainder of this paper is structured as follows. In Section \ref{sec:convergence}, we establish the linear convergence of the gradient descent method with a two-stage initialization to the global minimizer. Subsequently, Section \ref{sec:numerical} presents the results of numerical experiments. To further support our findings,  Sections \ref{sec:local_curvature_smoothness} and \ref{sec:proof_theorem} contain the detailed proofs of theorems and Section \ref{sec:concentration} presents concentration inequalities based on those i.i.d. rank-$r$ orthogonal projections in $\C^d$ drawn from the Haar measure with $1\le r<d$. Additional supplementary information can be found in the appendix. 
	%########################################################################
	\section{Convergence Analysis}
	\label{sec:convergence}
	Throughout this paper, $A_1,\ldots,A_N$ are i.i.d. rank-$r$ orthogonal projections drawn from the Haar measure, where $1\le r<d$, and $y_j=\vx^*A_j\vx,j=1,\ldots,N$. Define
	\[
	\varepsilon_0=\sqrt{\frac{10\zeta}{27\lambda}},\qquad
	\varepsilon_1=\sqrt{\frac{10}{27}},
	\]
	where  
	\begin{equation*}
		\lambda=\frac{r(r+1)}{d(d+1)},\qquad \mu=\frac{dr^2-r}{d(d+1)(d-1)},\qquad
		\zeta=\lambda-\mu=\frac{r(d-r)}{d(d+1)(d-1)}.
	\end{equation*}
	These parameters satisfy
	\[
	0<\zeta\le\mu<\lambda\le1,
	\qquad
	\mu-\zeta=\frac{r(r-1)}{d(d-1)}\ge0.
	\]
	We consider
	\begin{equation}\label{model:PR-optimization}
		E(\vz)=\frac{1}{2N}\sum_{j=1}^N(\vz^*A_j\vz-y_j)^2,\ \widetilde E(\vz)=\frac{1}{2N}\sum^{N}_{j=1}\left(\vz^* A_j\vz-y_j\right)^2-\frac\mu2\left(\|\vz\|^2-\frac{d}{rN}\sum^N_{j=1}y_j\right)^2,
	\end{equation}
	and the corresponding Wirtinger gradients about $\vz$ are
	\[
	\nabla_{\vz} E(\vz)=\frac1N\sum_{j=1}^N(\vz^*A_j\vz-y_j)A_j\vz,\qquad\nabla_{\vz} \widetilde E(\vz)=\frac1N\sum_{j=1}^N(\vz^*A_j\vz-y_j)A_j\vz-\mu\left(\|\vz\|^2-\frac{d}{rN}\sum^N_{j=1}y_j\right)\vz.
	\] 
	
	\begin{definition}
		For any $\vz\in\C^d$, let $\theta(\vz)\in\arg\min_{\theta\in[0,2\pi)}\|\vz-\vx e^{i\theta}\| $
		and define $
		\dist(\vz,\vx)=\|\vz-\vx e^{i\theta(\vz)}\| $. 
		For any $\varepsilon>0$, set $
		\calS(\vx,\varepsilon)=\{\vz\in\C^d:\dist(\vz,\vx)\le\varepsilon\|\vx\| \}$.
	\end{definition} 
	
	\begin{theorem}[Local linear convergence]
		\label{theo:convergence}
		Let $0<\delta<1$ and let $\vz_0\in\calS(\vx,\varepsilon_0)$. There exist constants $\gamma_{0},C_{0,\delta},c_{0,\delta}>0$ such that, if
		\begin{equation}\label{eq:N2}
			N\ge C_{0,\delta}d\left(\frac{\lambda}{\zeta}\right)^2\log^2d,
		\end{equation}
		with probability at least
		$
		1-7e^{-c_{0,\delta}d}-e^{-\gamma_{0}d\log^2d}
		-Nd^{-7r}$, 
		the iterates generated by gradient descent on $E$ satisfy
		\begin{equation}\label{eq:local_linear_rate}
			\dist^2(\vz_{k+1},\vx)\le
			\left(1-\frac{\xi\zeta(1-\delta)}4\|\vx\| ^2\right)
			\dist^2(\vz_k,\vx),\qquad k\ge0,
		\end{equation}
		where $\xi$ is the step size, satisfying
		\[
		0<\xi\le\frac{2\zeta}{L_{0,\delta}\lambda\|\vx\|^2},
		\]
		with $L_{0,\delta}=\max\left\{120\frac{(1+\delta)^2}{1-\delta},270(1+\delta)\right\}$. In particular, $\vz_k\in\calS(\vx,\varepsilon_0)$ for every $k\ge0$.
	\end{theorem}
	
	\begin{proof}
		The proof is deferred to Subsection \ref{subsec:proof-convergence}.
	\end{proof}
	%########################################################################
	From Theorem \ref{theo:convergence}, we require the initial guess to be $\vz_0\in\calS(\vx,\varepsilon_0)$ with $\varepsilon_0 = \sqrt{\frac{10\zeta}{27\lambda}}$ to guarantee convergence. In the following, we explore the two-stage initialization in Subsection \ref{subsec:2stage}.
	\subsection{Two-Stage Initialization}\label{subsec:2stage}
	\begin{algorithm}[ht]
		\caption{Two-stage initialization}
		\label{alg:two_stage_initialization}
		\begin{algorithmic}[1]
			\Require Measurements $\{y_j\}_{j=1}^N$, projections $\{A_j\}_{j=1}^N$, step size $\widetilde \xi>0$, and iteration number $B$.
			\State Form $Y=N^{-1}\sum_{j=1}^Ny_jA_j$ and $\rho^2=d(rN)^{-1}\sum_{j=1}^Ny_j$. Here $\rho$ denotes the nonnegative square root. 
			\State Let $\widehat\vz_0$ be a unit leading eigenvector of $Y$, and set $\widetilde\vz_0=\rho\widehat\vz_0$.
			\State Define 
			\[\widetilde E(\vz)=(2N)^{-1}\sum_{j=1}^N(\vz^*A_j\vz-y_j)^2-\mu(\|\vz\| ^2-\rho^2)^2/2.\]
			\For{$b=0,\ldots,B-1$}
			\State $\widetilde\vz_{b+1}=\widetilde\vz_b-\widetilde\xi\nabla_{\vz}\widetilde E(\widetilde\vz_b)$.
			\EndFor
			\State \Return the trajectory $\widetilde\vz_0,\ldots,\widetilde\vz_B$.
		\end{algorithmic}
	\end{algorithm}
	
	In this subsection, we present a two-stage initialization algorithm in Algorithm \ref{alg:two_stage_initialization}. Since $\E y_j=(r/d)\|\vx\|^2$, the quantity $\rho$ estimates $\|\vx\|$. The correction term in $\widetilde E$ removes the isotropic part of the population curvature and enlarges the region in which a descent estimate is available.
	\begin{theorem}
		\label{theo:relation}
		For any $\vz$ satisfying $\dist(\vz,\vx)\in[\frac{\varepsilon_0}2\|\vx\|,\varepsilon_1\|\vx\|]$, set $\widetilde \vz=\vz-\widetilde\xi\nabla_{\vz} \widetilde E(\vz)$.
		For any $0<\delta<1$, there exist constants $\widetilde\gamma_0>0$ and $\widetilde c_{0,\delta},\widetilde C_{0,\delta}>0$ depending on $\delta$, 
		when 
		\begin{equation}\label{eq:N_tilde}
			N\ge \widetilde C_{0,\delta}\frac{d^2}{d-r}\left(\frac{\lambda}{\zeta}\right)^2\log^2d,
		\end{equation}  
		we have 
		\[\dist^2(\widetilde\vz,\vx)\le\left(1-\frac{\widetilde\xi\zeta(1-\delta)}4\|\vx\|^2\right)\dist^2(\vz,\vx)\]
		with probability at least $1-7\exp(-\widetilde c_{0,\delta} d)-Nd^{-7r}
		-\exp(-\widetilde\gamma_0 d\log^2 d)$, where $\widetilde L_{0,\delta}=\max\left\{\frac{16\phi}{1-\delta},540(1+\delta)\right\}$ and
		$0<\widetilde\xi\le \frac{2\zeta}{\widetilde L_{0,\delta}\lambda\|\vx\|^2}$ with $		\phi=15(1+\delta)^2+(1+\varepsilon_1)^2(2+\varepsilon_1+\delta)^2$.
		
	\end{theorem}
	\begin{proof}
		The proof is given in Subsection \ref{subsec:proof_convergence1}.
	\end{proof}
	\begin{theorem}[Invariance of the regularized inner ball]\label{theo:regularized_inner_invariance}
		For any $0<\delta<1$ and any $\vz\in\calS(\vx,\varepsilon_0)$, there exist constants $\gamma_1>0$ and $ c_{1,\delta}, C_{1,\delta}>0$ depending on $\delta$,  if
		\begin{equation}\label{eq:inner-step-size}
			N\ge  C_{1,\delta}\frac{d^2}{d-r}\left(\frac{\lambda}{\zeta}\right)^2\log^2d,\ 0<\widetilde\xi\le\frac{2\zeta}{ L_{1,\delta}\lambda\|\vx\|^2},
		\end{equation}
		where $ L_{1,\delta}=\max\left\{\frac{16\phi}{1-\delta},540(1+\delta),4\left[
		L_{0,\delta}
		\left(
		\frac{1-\delta}{32}
		+\frac{1+\delta}{432}
		\right)
		\right]^{1/2}
		+
		2\left(2+\frac{\varepsilon_1}{2}+\delta\right)
		\left(1+\frac{\varepsilon_1}{2}\right)\right\}$ with $		\phi=15(1+\delta)^2+(1+\varepsilon_1)^2(2+\varepsilon_1+\delta)^2$ and $L_{0,\delta}=\max\left\{120\frac{(1+\delta)^2}{1-\delta},270(1+\delta)\right\}$, then
		$
		\vz-\widetilde\xi\nabla_{\vz}\widetilde E(\vz)\in\calS(\vx,\varepsilon_0)$ holds with probability at least $1-8\exp(- c_{1,\delta} d)-Nd^{-7r}
		-\exp(-\gamma_1 d\log^2 d)$.
	\end{theorem}
	\begin{proof}
		The proof is given in Subsection \ref{subsec:proof-invariance}.
	\end{proof}
	\begin{theorem}[Two-stage initialization]
		\label{theo:good_initialization_2stage}
		Let $0<\delta<1$. Fix a step size satisfying
		\[
		0<\widetilde\xi\le\frac{2\zeta}{ L_{1,\delta}\lambda\|\vx\|^2},
		\]
		and define the deterministic iteration count
		\begin{equation}\label{eq:B-xi}
			B(\widetilde\xi):=\left\lceil\frac{4}{\widetilde\xi\zeta(1-\delta)\|\vx\|^2}\log\!\left(\frac{\lambda}{\zeta}\right)\right\rceil,
		\end{equation}
		where $L_{1,\delta}$ is defined in Theorem \ref{theo:regularized_inner_invariance} and   $\lceil\cdot\rceil$ is to get the smallest upper integer bound. There exist constants $\widetilde C_{1,\delta},\widetilde c_{1,\delta},\widetilde\gamma_1>0$ such that, if
		\[
		N\ge \widetilde C_{1,\delta}\frac{d^2}{d-r}\left(\frac{\lambda}{\zeta}\right)^2\log^2d,
		\]
		then the spectral initializer in Line 2 of Algorithm~\ref{alg:two_stage_initialization} satisfies $\widetilde\vz_0\in\calS(\vx,\varepsilon_1)$ and the two-stage initializer $\widetilde\vz_{B(\widetilde\xi)}\in\calS(\vx,\varepsilon_0)$, with probability at least $1-8e^{-\widetilde c_{1,\delta}d}-e^{-\widetilde\gamma_1d\log^2d}-Nd^{-7r}$.  In particular, if $\widetilde\xi=c_\delta\zeta/(\lambda\|\vx\|^2)$ for a sufficiently small constant $c_\delta>0$, then
		$B(\widetilde\xi)=O_\delta\!\left(\frac{\lambda}{\zeta^2}\log\frac{\lambda}{\zeta}\right)$.
	\end{theorem}
	\begin{proof}
		The proof is given in Subsection~\ref{subsec:proof_2stage_spectral_init}.
	\end{proof}
	Theorem~\ref{theo:good_initialization_2stage} shows that the regularized trajectory enters the local convergence region within $B(\widetilde\xi)$ iterations. Combining Theorems~\ref{theo:convergence} and \ref{theo:good_initialization_2stage}, we have the following overall convergence.
	
	\begin{theorem}[Overall convergence]
		\label{theo:overall_convergence}
		We generate the regularized trajectory $\{\widetilde{\vz}_b\}_{b=0}^B$ according to Algorithm~\ref{alg:two_stage_initialization}, with $B=B(\widetilde\xi)$ defined in \eqref{eq:B-xi}, and then initialize the gradient descent scheme for objective function $E$ in \eqref{model:PR-optimization} with $\vz_0 = \widetilde{\vz}_B$. Let $0<\delta<1$  be fixed, and assume that the sample size and probability conditions stated in Theorems~\ref{theo:convergence} and~\ref{theo:good_initialization_2stage} are satisfied, we have 
		\[
		\dist^2(\vz_k,\vx)\le\left(1-\frac{\xi\zeta(1-\delta)}4\|\vx\|^2\right)^k\dist^2(\vz_0,\vx),\qquad k\ge0,
		\]
		where $\xi$ is the step size used for the gradient descent scheme on $E$, satisfying
		$
		0<\xi\le
		\frac{2\zeta}{L_{0,\delta}\|\vx\|^2\lambda}$. 
		In particular, when $r=O(1)$, the required number of measurements is $O(d\log^2d)$.
	\end{theorem}	
	%########################################################################
	\section{Numerical Experiments}
	\label{sec:numerical}
	We present some numerical experiments to evaluate our proposed algorithm here. We test the two-stage initialization for the fusion frame phase retrieval and explore the minimal measurement number we need for successful recovery with different ranks.
	
	\subsection{Experimental Setup}
	In the following numerical evaluations, we employ gradient descent with the two-stage initialization to recover signals from the random measurements. The target signal, denoted as $\vx=[x_1\ \ldots\ x_d]^\top \in\C^d$, is randomly generated from a Gaussian distribution.
	
	The sampling matrices utilized in our approach are random rank-$r$ orthogonal projection matrices. Specifically, we construct a matrix $M_j\in\C^{r\times d}$, where the entries are i.i.d. random variables following a Gaussian distribution $\frac{1}{\sqrt{2}}[\calN(0,1)+i\calN(0,1)]$ in the complex case. The sampling matrix is then obtained as $A_j=M_j^*(M_jM^*_j)^{-1}M_j$.
	
	In the iterative process, we employ Barzilai-Borwein's method (B-B method) to determine the step size. This approach utilizes information from the previous iteration to compute the step size for the current iteration. Specifically, to obtain the $(k+1)$-th solution $\vz_{k+1}=\vz_k-\xi_k\nabla_{\vz} E(\vz_k)$, we choose
	\begin{equation*}
		\xi_k=\underset{\xi}{\arg\min} \,\,\|\vs_k-\xi\cdot\vg_k\|^2,
	\end{equation*}
	where $\vs_k=\vz_k-\vz_{k-1}$ and $\vg_k=\nabla_{\vz} E(\vz_k)-\nabla_{\vz} E(\vz_{k-1})$.
	By performing some simple calculations, we obtain
	$$
	\xi_k=\frac{\re(\langle \vg_k,\,\vs_k\rangle)}{\|\vg_k\|^2}=\frac{ \re(\vg_k^*\vs_k)}{\|\vg_k\|^2}.
	$$
	To ensure that the step size is positive, we take the absolute value of $\xi_k$ as the step size. The iterative algorithm terminates if $\|\vg_k\| < 10^{-16}$ or if the maximum iteration number is reached.
	
	To assess the performance of our method, we use the relative error of the reconstruction, defined as the distance between the numerical solution $\vz$ and the true solution $\vx$, divided by the norm of $\vx$.
	%%%%%%%%%%%
	\subsection{Numerical Results}
	\label{subsec:numerical_result}
	\begin{figure}[htbp]
		\begin{center}
			\includegraphics[width=0.48\textwidth]{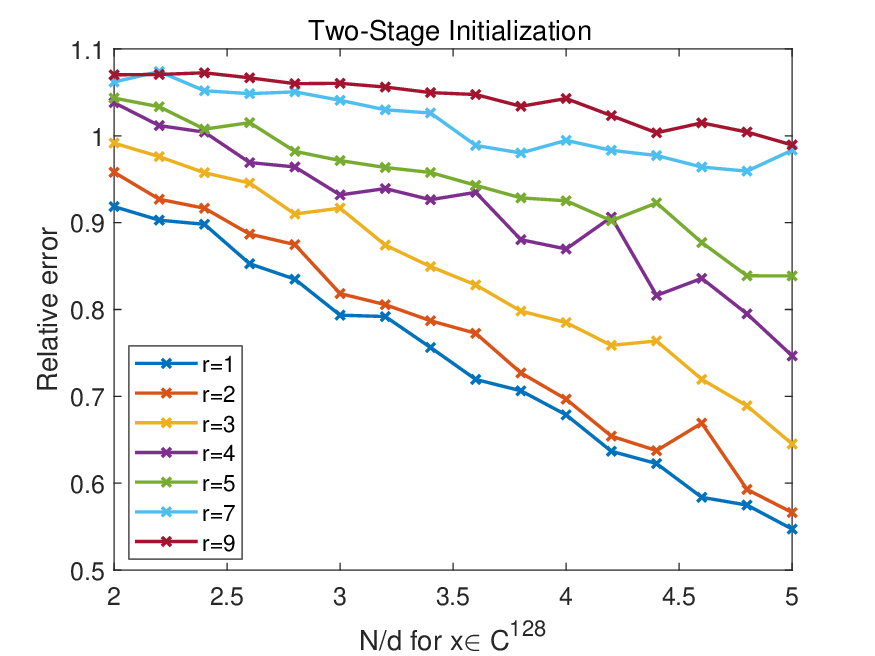}
			\includegraphics[width=0.48\textwidth]{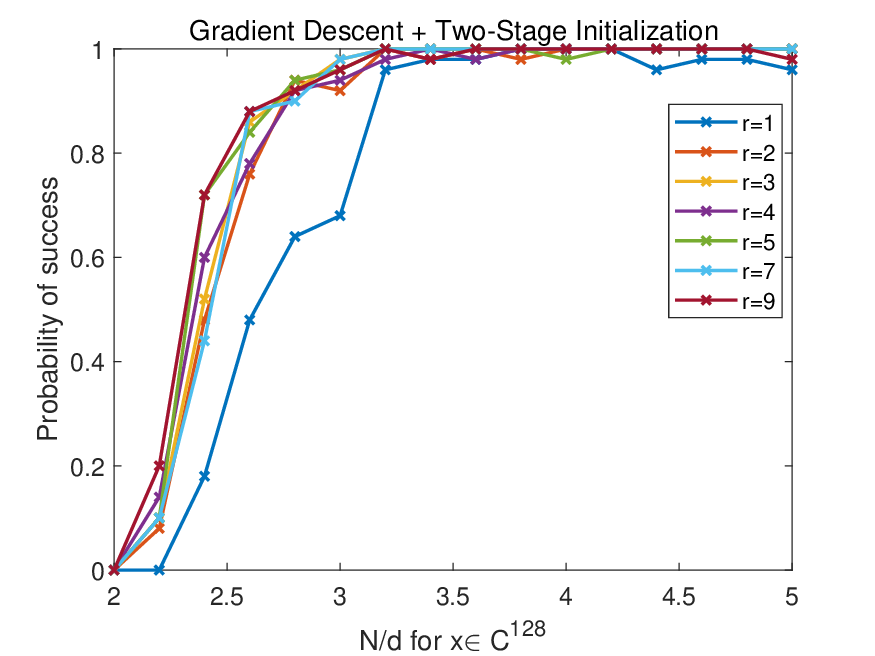}
			\caption{Left: Averaged relative error between $ \vz $ and $ \vx $ for initialization experiments. Right: The success rates versus $N/d$.} \label{fig:init1}
		\end{center}
	\end{figure}
	
	To evaluate the effectiveness of the two-stage initialization for different ranks, we investigate the relationship between the relative error and $N/d$ in fusion frame phase retrieval. In our experiments, we fix $d$ at a value of $128$ and vary $N/d$ and $r$ within the ranges of $[2:0.2:5]$ and $[1:5,7,9]$ respectively.
	
	For each combination of $N/d$ and $r$, we conduct 50 repetitions of the experiment and compute the average relative error. The left subfigure in
	Figure \ref{fig:init1} depicts the plot of relative errors against $N/d$ for different values of $r$. It is evident that the averaged relative error does not improve with increasing rank-$r$ in the two-stage initialization.

	Next, our focus shifts to evaluating the overall algorithm for solving the fusion frame phase retrieval problem. We set the maximum iteration number to $2000$ and conduct $100$ trials to assess the performance of our algorithm.
	
	During each trial, we consider it a success if the relative error of the reconstruction falls below $10^{-5}$. The empirical probability of success is then calculated as the average success rate across the $100$ trials.
	
	In these experiments, we select a value of $d = 128$. Figure \ref{fig:init1} (the right subfigure) displays the average rates as a function of $N/d$, considering various values of $r$. The plot clearly illustrates that the highest number of measurements is required when $r=1$, which implies the increase of the value $r$ can improve the probability of success. 
	Within the tested range $1\le r\le9$, the empirical success probability generally shows better performance with $r>1$ for a fixed value of $N/d$.
	
	%###########################################
	\section{Local Curvature Conditions and Local Smoothness Conditions}\label{sec:local_curvature_smoothness}	
	\begin{lemma}[Local Curvature Conditions]	\label{lem:lem_regularity_cond}
		Under the same assumptions as Theorem \ref{theo:convergence}, for every $0<\delta<1$, there exist positive constants $\gamma_2>0$ and $c_{2,\delta}, C_{2,\delta}>0$ such that the following statements hold with probability at least $1-7\exp(-c_{2,\delta} d)-Nd^{-7r}
		-\exp(-\gamma_2 d\log^2 d)$\\
		(1) for any $ \vz\in\calS(\vx,\varepsilon_0)$, when $N\ge C_{2,\delta} d\left(\frac{\lambda}{\zeta}\right)^2\log^2 d$:
		\begin{equation}\label{curvature_condition0}
			\begin{split}
				&\re\left(\langle\nabla_{\vz} E(\vz),\vz-\vx e^{i\theta(\vz)}\rangle\right)
				\ge \frac{\zeta(1-\delta)}{4}\|\vx\|^2\cdot\dist^2(\vz,\vx)+\frac{1}{10N}\sum_{j=1}^{N}\big|(\vz-\vx e^{i\theta(\vz)})^*A_j(\vz-\vx e^{i\theta(\vz)})\big|^2,
			\end{split}
		\end{equation}
		(2) for all $\vz$ satisfying $\dist(\vz,\vx)\in[\frac{\varepsilon_0}2\|\vx\|,\varepsilon_1\|\vx\|]$, when $N\ge C_{2,\delta}\max\left(\frac dr\left(\frac{\mu}{\zeta}\right)^2\frac{\lambda}{\zeta}, d\left(\frac{\lambda}{\zeta}\right)^2\log^2 d\right)$:  
		
		\begin{equation}\label{eq:curvature_condition0}
			\begin{split}
				\re\left(\langle\nabla_{\vz} \widetilde E(\vz),\vz-\vx e^{i\theta(\vz)}\rangle\right)\ge& \frac{\zeta(1-\delta)}{4}\|\vx\|^2\cdot\dist^2(\vz,\vx)+\frac{\zeta}{10\lambda N}\sum_{j=1}^{N}\big|(\vz-\vx e^{i\theta(\vz)})^*A_j(\vz-\vx e^{i\theta(\vz)})\big|^2.
			\end{split}
		\end{equation}
		Here, $\varepsilon_0 =\sqrt{\frac{10\zeta}{27\lambda}}$, $\varepsilon_1 =\sqrt{\frac{10}{27}}$ and the values of $\lambda$ and $\zeta$ are defined in \eqref{defi_para}.
	\end{lemma}
	\begin{proof} 
		Let $\vh=e^{-i\theta(\vz)}\vz-\vx$. 
		By Lemma \ref{lem:concentration1-1}  for $ N $ satisfying \eqref{N}, with probability at least $1-6\exp(-c_{2,\delta} d)-Nd^{-7r}$, we have
		\begin{equation}
			\label{eq:real_part}
			\frac{1}{N}\sum^N_{j=1}\big(\re(\vh^*A_j\vx)\big)^2
			\le \E\big(\re^2(\vh^*A_j\vx)\big)+\frac{\delta\zeta}{2}\|\vx\|^2\|\vh\|^2
		\end{equation}
		for any $ \vh\in\C^d $. By Theorem \ref{theo:concentration1}, for every $\delta\ge0$, there exist constants $c_{2,\delta},C_{2,\delta}>0$, when $N$ satisfies \eqref{eq:N_constant}, we have 
		\begin{equation}\label{eq:concentration_bound}
			\Big\|\frac1N\sum_{j=1}^N A_j-\frac rd I_d\Big\|\le \delta\frac rd 
		\end{equation}
		with probability at least $1-\exp(-c_{2,\delta}d)$. Then, \eqref{eq:concentration_bound} implies Inequality \eqref{curvature_condition1} in Lemma \ref{lem:square_estiate} 
		\begin{equation}\label{eq:temp}
			\frac{1}{N}\sum^N_{j=1}\Bigg( \sqrt{\frac{5}{2}}\re(\vh^*A_j\vx)+\sqrt{\frac{9}{10}}|\vh^*A_j\vh|\Bigg)^2\ge\frac{\zeta}{4}\|\vx\|^2\|\vh\|^2+\frac{ \E\big(\re^2(\vh^*A_j\vx)\big)}{2},\ 0< \|\vh\|\le\varepsilon_0\|\vx\|
		\end{equation}
		with probability at least $1-\exp(-\gamma_2 d\log^2d)$. Combining \eqref{eq:real_part} and \eqref{eq:temp}, we have
		\begin{align*}
			&\re\left(\langle\nabla_{\vz} E(\vz),\,\vz-\vx e^{i\theta(\vz)}\rangle\right)\\
			=&\frac{1}{N}\sum^N_{j=1}\Bigg( \sqrt{\frac{5}{2}}\re(\vh^*A_j\vx)+\sqrt{\frac{9}{10}}|\vh^*A_j\vh|\Bigg)^2-\frac{1}{2N}\sum^N_{j=1}\big(\re(\vh^*A_j\vx)\big)^2+\frac{1}{10N}\sum^N_{j=1}|\vh^*A_j\vh|^2\\
			\ge&\frac{\zeta}{4}\|\vx\|^2\|\vh\|^2+\frac{ \E\big(\re^2(\vh^*A_j\vx)\big)}{2}-\frac{1}{2N}\sum^N_{j=1}\big(\re(\vh^*A_j\vx)\big)^2+\frac{1}{10N}\sum^N_{j=1}|\vh^*A_j\vh|^2\\
			\ge&\frac{\zeta}{4}\|\vx\|^2\|\vh\|^2+\frac{ \E\big(\re^2(\vh^*A_j\vx)\big)}{2}-\frac12\E\big(\re^2(\vh^*A_j\vx)\big)-\frac{\delta\zeta}{4}\|\vx\|^2\|\vh\|^2+\frac{1}{10N}\sum^N_{j=1}|\vh^*A_j\vh|^2\\
			=&\frac{1}{10N}\sum^N_{j=1}|\vh^*A_j\vh|^2+\frac{(1-\delta)\zeta}{4}\|\vx\|^2\|\vh\|^2,
		\end{align*}
		which corresponds to Inequality \eqref{curvature_condition0}.
		
		Next, we will prove Inequality \eqref{eq:curvature_condition0}. 
		Let $\vh=e^{-i\theta(\vz)}\vz-\vx$, then $\|\vz\|\le(1+\varepsilon_1)\|\vx\|$.   
		Recall that $\rho^2=d(rN)^{-1}\sum_{j=1}^Ny_j$, Theorem \ref{theo:concentration1} and the lower bound $\|\vh\|\ge\varepsilon_0\|\vx\|/2$ give
		\begin{equation}\label{eq:concentration_rho}
			\left|\|\vx\|^2-\rho^2\right|
			\le\frac{\hat\delta\zeta\|\vx\|}{4(1+\varepsilon_1)\mu}
			\frac{\varepsilon_0}{2}\|\vx\|
			\leq\frac{\hat\delta\zeta}{4(1+\varepsilon_1)\mu}\|\vx\|\,\|\vh\|.
		\end{equation}
		with probability at least $1-\exp(-c_{2,\delta} d)$.  
		Note that
		\begin{align}
			&\left(\|\vz\|^2-\frac{d}{rN}\sum^N_{j=1}y_j\right)\re(\vh^*\vz e^{-i\theta(\vz)})=\left(\|\vz\|^2-\|\vx\|^2+\|\vx\|^2-\frac{d}{rN}\sum^N_{j=1}y_j\right)\re(\vh^*\vz e^{-i\theta(\vz)})\nonumber\\
			\le &(\|\vz\|^2-\|\vx\|^2)\re(\vh^*\vz e^{-i\theta(\vz)})+\left|\|\vx\|^2-\frac{d}{rN}\sum^N_{j=1}y_j\right|\|\vh\|\|\vz\|\nonumber\\
			\le &(\|\vz\|^2-\|\vx\|^2)\re(\vh^*\vz e^{-i\theta(\vz)})+\frac{\hat\delta\zeta}{4\mu}\|\vh\|^2\|\vx\|^2
			=(\|\vz\|^2-\|\vx\|^2)\re(\vh^*(\vx+\vh))+\frac{\hat\delta\zeta}{4\mu}\|\vh\|^2\|\vx\|^2\nonumber\\=&\|\vh\|^4+3\|\vh\|^2\re(\vx^*\vh)+2(\re(\vx^*\vh))^2+\frac{\hat\delta\zeta}{4\mu}\|\vh\|^2\|\vx\|^2,\label{eq:estimate_regular}
		\end{align}
		where the second inequality is from \eqref{eq:concentration_rho} since
		$\left|\|\vx\|^2-\rho^2\right|\|\vh\|\|\vz\|
		\le(\hat\delta\zeta/(4\mu))\|\vx\|^2\|\vh\|^2$, which is exactly the bound used in \eqref{eq:estimate_regular}.
		By Lemma \ref{lem:concentration1-1}  for $ N $ satisfying \eqref{N}, we have
		\begin{equation}
			\label{eq:real_part_tilde}
			\frac{1}{N}\sum^N_{j=1}\big(\re(\vh^*A_j\vx)\big)^2
			\le \E\big(\re^2(\vh^*A_j\vx)\big)+\frac{\widetilde\delta\zeta}{2}\|\vx\|^2\|\vh\|^2
		\end{equation}
		for any $ \vh\in\C^d $ and any $\widetilde\delta>0$ with probability at least $1-6\exp(-c_{2,\delta} d)-Nd^{-7r}$.
		\begin{align*}
			&\re\left(\langle\nabla_{\vz} \widetilde{E}(\vz),\,\vz-\vx e^{i\theta(\vz)}\rangle\right)\\
			=&\frac{1}{N}\sum^N_{j=1}\left(2\big(\re(\vh^*A_j\vx)\big)^2
			+3\re(\vh^*A_j\vx)(\vh^*A_j\vh)+|\vh^*A_j\vh|^2\right)-\mu\left(\|\vz\|^2-\frac{d}{Nr}\sum^N_{j=1}y_j\right)\re(\vh^*\vz e^{-i\theta(\vz)})\\
			=&\frac{1}{N}\sum_{j=1}^{N}\Bigg( \frac{3}{2\sqrt{1-\frac{\zeta}{10\lambda}}}\re(\vh^*A_j\vx)+\sqrt{1-\frac{\zeta}{10\lambda}}|\vh^*A_j\vh|\Bigg)^2+\frac{\zeta}{10\lambda N}\sum_{j=1}^{N}\big|\vh^*A_j\vh\big|^2\\
			&-\mu\left(\|\vz\|^2-\frac{d}{Nr}\sum^N_{j=1}y_j\right)\re(\vh^*\vz e^{-i\theta(\vz)})-\left(\frac{9}{4\left(1-\frac{\zeta}{10\lambda}\right)}-2\right) \frac{1}{N}\sum_{j=1}^{N}\re^2(\vh^*A_j\vx)\\
			\ge&\frac{1}{N}\sum_{j=1}^{N}\Bigg( \frac{3}{2\sqrt{1-\frac{\zeta}{10\lambda}}}\re(\vh^*A_j\vx)+\sqrt{1-\frac{\zeta}{10\lambda}}|\vh^*A_j\vh|\Bigg)^2+\frac{\zeta}{10\lambda N}\sum_{j=1}^{N}\big|\vh^*A_j\vh\big|^2\\
			&-\mu(\|\vh\|^4+3\re(\vx^*\vh)\|\vh\|^2+2(\re(\vx^*\vh))^2)-\frac{\hat\delta\zeta}{4}\|\vh\|^2\|\vx\|^2-\left(\frac{9}{4\left(1-\frac{\zeta}{10\lambda}\right)}-2\right)\E\big(\re^2(\vh^*A_j\vx)\big)\\
			&-\left(\frac{9}{4\left(1-\frac{\zeta}{10\lambda}\right)}-2\right)\cdot\frac{\widetilde\delta\zeta}{2}\|\vx\|^2\|\vh\|^2\\
			\ge&\left(\frac{9}{4\left(1-\frac{\zeta}{10\lambda}\right)}-2\right)\E\big(\re^2(\vh^*A_j\vx)\big)+\frac{\zeta}{4}\|\vx\|^2\|\vh\|^2-\frac{\hat\delta\zeta}{4}\|\vh\|^2\|\vx\|^2+\frac{\zeta}{10\lambda N}\sum_{j=1}^{N}\big|\vh^*A_j\vh\big|^2\\
			&-\left(\frac{9}{4\left(1-\frac{\zeta}{10\lambda}\right)}-2\right)\E\big(\re^2(\vh^*A_j\vx)\big)-\left(\frac{9}{4\left(1-\frac{\zeta}{10\lambda}\right)}-2\right)\cdot\frac{\widetilde\delta\zeta}{2}\|\vx\|^2\|\vh\|^2\\=&\frac{\zeta}{10\lambda N}\sum_{j=1}^{N}\big|\vh^*A_j\vh\big|^2+\left(1-\hat\delta-2\left(\frac{9}{4\left(1-\frac{\zeta}{10\lambda}\right)}-2\right)\widetilde\delta\right)\frac{\zeta}{4}\|\vx\|^2\|\vh\|^2\\:=&\frac{\zeta}{10\lambda N}\sum_{j=1}^{N}\big|\vh^*A_j\vh\big|^2+\frac{(1-\delta)\zeta}{4}\|\vx\|^2\|\vh\|^2,  
		\end{align*} 
		where the first inequality is from \eqref{eq:estimate_regular} and \eqref{eq:real_part_tilde}, the second inequality is from Inequality \eqref{curvature_condition2} in Lemma \ref{lem:square_estiate} with probability at least $1-\exp(-\gamma_2 d\log^2d)$.
		In the last line, regard $\delta\in(0,1)$ as the prescribed final accuracy, put
		\[
		C_*:=2\left(\frac{9}{4(1-\zeta/(10\lambda))}-2\right)>0,
		\qquad \hat\delta:=\frac\delta2,
		\qquad \widetilde\delta:=\frac{\delta}{2C_*}.
		\]
		Then $1-\hat\delta-C_*\widetilde\delta=1-\delta$, as required.	
	\end{proof}
	\begin{lemma}[Local Smoothness Conditions]\label{lem:lem_smooth_cond}
		Under the same assumptions as in Theorem \ref{theo:convergence}, there exist constants $c_{3,\delta}, C_{3,\delta}>0$, the following inequalities hold\\
		(1) when $ N$ satisfies $N\ge C_{3,\delta} d\left(\frac{\lambda}{\zeta}\right)^2\log^2 d$, with probability at least $1-7\exp(-c_{3,\delta} d)-Nd^{-7r}$, we have  for all $\vz\in\calS(\vx,\varepsilon_0)$
		\begin{align*}	
			\|\nabla_{\vz} E(\vz)\|^2
			\leq& L_{0,\delta}\frac{\lambda}{\zeta}\|\vx\|^2\Bigg(\frac{\zeta(1-\delta)}{8}\|\vx\|^2\dist^2(\vz, \vx)+\frac{1}{10N}\sum_{j=1}^{N}\big|(\vz-\vx e^{i\theta(\vz)})^*A_j(\vz-\vx e^{i\theta(\vz)})\big|^2\Bigg),
		\end{align*}
		(2) when $N\ge C_{3,\delta}\max\left(\frac dr\left(\frac{\mu}{\zeta}\right)^2\frac{\lambda}{\zeta}, d\left(\frac{\lambda}{\zeta}\right)^2\log^2 d\right)$, with probability at least $1-7\exp(-c_{3,\delta} d)-Nd^{-7r}$, we have  for all $\vz$ satisfying $\dist(\vz,\vx)\in[\frac{\varepsilon_0}2\|\vx\|,\varepsilon_1\|\vx\|]$  
		\begin{equation*}	
			\|\nabla_{\vz} \widetilde E(\vz)\|^2
			\leq \widetilde L_{0,\delta}\frac{\lambda}{\zeta}\|\vx\|^2\Bigg(\frac{\zeta(1-\delta)}{8}\|\vx\|^2\dist^2(\vz, \vx)+\frac{\zeta}{10\lambda N}\sum_{j=1}^{N}\big|(\vz-\vx e^{i\theta(\vz)})^*A_j(\vz-\vx e^{i\theta(\vz)})\big|^2\Bigg).
		\end{equation*}
		Here $\varepsilon_0=\sqrt{\frac{10\zeta}{27\lambda}}$, $\varepsilon_1=\sqrt{\frac{10}{27}}$, and
		\[
		L_{0,\delta}=\max\left\{120\frac{(1+\delta)^2}{1-\delta},270(1+\delta)\right\},\ \widetilde L_{0,\delta}=\max\left\{\frac{16\phi}{1-\delta},540(1+\delta)\right\},
		\
		\phi=15(1+\delta)^2+(1+\varepsilon_1)^2(2+\varepsilon_1+\delta)^2.
		\]
	\end{lemma}
	%--------------------------------------------------------------------------------------------------------------------
	\begin{proof}  We prove the result for $E$. Set $\vh:=e^{-i\theta(\vz)}\vz-\vx$, then $\|\vh\|\le\varepsilon_0\|\vx\|$. For any $\vu\in\C^d$ with $\|\vu\|=1$, let $\vv=e^{-i\theta(\vz)}\vu$, we calculate
		\begin{align*}
			& |\nabla_{\vz} E(\vz)^*\vu|^2
			=\Big|\frac{1}{N}\sum^N_{j=1}\Big(2\re(\vx^*A_j\vh)+\vh^*A_j\vh\Big)\Big(\vv^*A_j\vx+\vv^*A_j\vh\Big)\Big|^2\\
			\le&\bigg(\frac{1}{N}\sum^N_{j=1}3\sqrt{|\vx^*A_j\vx|\,|\vv^*A_j\vv|}\,|\vh^*A_j\vh|+2\sqrt{|\vh^*A_j\vh|\,|\vv^*A_j\vv|}\,|\vx^*A_j\vx|+\sqrt{|\vh^*A_j\vh|^3|\vv^*A_j\vv|}\bigg)^2\\
			\le& 27\,\Big(\frac{1}{N}\sum^N_{j=1}\sqrt{|\vx^*A_j\vx|\,|\vv^*A_j\vv|}\,|\vh^*A_j\vh|\Big)^2+12\,\Big(\frac{1}{N}\sum^N_{j=1}\sqrt{|\vh^*A_j\vh|\,|\vv^*A_j\vv|}\,|\vx^*A_j\vx|\Big)^2\\
			&+3\,\Big(\frac{1}{N}\sum^N_{j=1}\sqrt{|\vh^*A_j\vh|^3|\vv^*A_j\vv|}\Big)^2
			:= 27I_1+12I_2+3I_3.
		\end{align*}
		By Theorem \ref{theo:concentration1-1}, when $ N $ satisfies \eqref{N}, we have for any $\delta>0$, all $\vh\in\C^d$ and a fixed $\vx\in\C^d$
		\begin{equation}
			\label{eq:eq_Y}
			\vh^*\bigg(\frac{1}{N}\sum^N_{j=1}(\vx^*A_j\vx) A_j\bigg)\vh\le (1+\delta)\lambda\|\vh\|^2\|\vx\|^2
		\end{equation}
		holds with probability at least $1-6\exp(-c_{3,\delta} d)-Nd^{-7r}$. We apply Cauchy-Schwarz inequality together with \eqref{eq:eq_Y} to $I_1$ and $I_2$
		\begin{eqnarray*}
			I_1&\le&\bigg(\frac{1}{N}\sum^N_{j=1}|\vh^*A_j\vh|^2\bigg)\cdot\bigg(\frac{1}{N}\sum^N_{j=1}|\vv^*A_j\vv||\vx^*A_j\vx|\bigg)\\
			&\le&\frac{1}{N}\sum^N_{j=1}|\vh^*A_j\vh|^2\cdot\vv^*\bigg(\frac{1}{N}\sum^N_{j=1}(\vx^*A_j\vx) A_j\bigg)\vv
			\le(1+\delta)\lambda\|\vx\|^2\cdot\frac{1}{N}\sum^N_{j=1}|\vh^*A_j\vh|^2,\\
			I_2&\le&\vv^*\bigg(\frac{1}{N}\sum^N_{j=1}(\vx^*A_j\vx) A_j\bigg)\vv\cdot\vh^*\bigg(\frac{1}{N}\sum^N_{j=1}(\vx^*A_j\vx) A_j\bigg)\,\vh\le\lambda^2(1+\delta)^2\|\vx\|^4\|\vh\|^2.
		\end{eqnarray*}
		We estimate $I_3$ by $\|A_j\|\le 1,\forall j$ and $\frac{1}{N}\sum^N_{j=1}\vh^*A_j\vh\le(1+\delta)\frac{r}{d}\|\vh\|^2$ from Theorem \ref{theo:concentration1} with probability at least $1-e^{-c_{3,\delta}d}$, when $N$ satisfies \eqref{eq:N_constant}
		\begin{eqnarray*}
			I_3
			&\leq&\bigg(\frac{1}{N}\sum^N_{j=1}|\vh^*A_j\vh|^2\bigg)\cdot\bigg(\frac{1}{N}\sum^N_{j=1}|\vh^*A_j\vh||\vv^*A_j\vv|\bigg)
			\le \frac{(1+\delta)^2r^2}{d^2}\|\vh\|^6.
		\end{eqnarray*}
		When $\|\vh\|\le\varepsilon_0\|\vx\|\leq\varepsilon_1\|\vx\|$, we have $$I_3\le\frac{(1+\delta)^2r^2}{d^2}\|\vh\|^2\frac{\zeta^2}{\lambda^2}\|\vx\|^4\le(1+\delta)^2\lambda\|\vx\|^4\|\vh\|^2.$$ Set $L_{0,\delta}=\max\left(120\frac{(1+\delta)^2}{1-\delta},270(1+\delta)\right)$. Therefore, when $ N$ satisfies \eqref{N}, we have 
		\begin{align}
			&\|\nabla_{\vz} E(\vz)\|^2=\max_{\|\vu\|=1}|\nabla_{\vz} E(\vz)^*\vu|^2\nonumber\\
			\le& 27 \lambda(1+\delta)\|\vx\|^2\cdot\frac{1}{N}\sum^N_{j=1}|\vh^*A_j\vh|^2+12\lambda^2(1+\delta)^2\|\vx\|^4\|\vh\|^2+3\lambda(1+\delta)^2\|\vx\|^4\|\vh\|^2\nonumber\\
			\le&15\lambda(1+\delta)^2\|\vx\|^4\|\vh\|^2+\frac{27\lambda(1+\delta)\|\vx\|^2}{N}\sum^N_{j=1}|\vh^*A_j\vh|^2\nonumber\\ 
			\le& L_{0,\delta}\frac{\lambda}{\zeta}\|\vx\|^2\Bigg(\frac{\zeta(1-\delta)}{8}\|\vx\|^2\|\vh\|^2+\frac{1}{10N}\sum_{j=1}^{N}|\vh^*A_j \vh|^2\Bigg)\label{R_cond}
		\end{align}   
		with probability at least $1-7\exp(-c_{3,\delta} d)-Nd^{-7r}$.
		
		Next, we will prove the result for $\widetilde E$. Note that
		\begin{equation}\label{eq:bound_e_tilde}
			\|\nabla_{\vz} \widetilde E(\vz)\|^2=\left\|\nabla_{\vz} E(\vz)-\mu\left(\|\vz\|^2-\frac{d}{rN}\sum^N_{j=1}y_j\right)\vz\right\|^2\le2\|\nabla_{\vz} E(\vz)\|^2+2\mu^2\left\|\left(\|\vz\|^2-\frac{d}{rN}\sum^N_{j=1}y_j\right)\vz\right\|^2.
		\end{equation}
		Set $\vh:=e^{-i\theta(\vz)}\vz-\vx$, so that $\varepsilon_0\|\vx\|/2\le \|\vh\|\le\varepsilon_1\|\vx\|$.  Theorem \ref{theo:concentration1} implies
		\begin{equation}\label{eq:concentration_yj}
			\left|\|\vx\|^2-\rho^2\right|
			\le\frac{\delta\zeta\varepsilon_0}{8(1+\varepsilon_1)\mu}\|\vx\|^2
			\le\frac{\delta\varepsilon_0}{2}\|\vx\|^2
			\le\delta\|\vx\|\|\vh\|,
		\end{equation}
		with probability at least $1-\exp(-c_{3,\delta}d)$, where the second inequality uses $\zeta\le\mu$ and the last uses $\|\vh\|\ge\varepsilon_0\|\vx\|/2$. Note that $\|\vz\|^2-\rho^2=\|\vz\|^2-\|\vx\|^2+\|\vx\|^2-\rho^2$. Then 
		\begin{align}
			&\left\|\left(\|\vz\|^2-\frac{d}{rN}\sum^N_{j=1}y_j\right)\vz\right\|\le\left|\|\vz\|^2-\frac{d}{rN}\sum^N_{j=1}y_j\right|\|\vz\|
			\le\left(\|\vh\|(\|\vh\|+2\|\vx\|)+\left|\|\vx\|^2-\frac{d}{rN}\sum^N_{j=1}y_j\right|\right)\|\vz\|\nonumber\\
			\le&(1+\varepsilon_1)\left(2+\varepsilon_1+\delta\right)\|\vx\|^2\|\vh\|,\label{eq:bound_regular}
		\end{align}
		where the last inequality holds from $\|\vz\|=\|\vx+\vh\|\le\|\vx\|+\|\vh\|\le(1+\varepsilon_1)\|\vx\|$ and \eqref{eq:concentration_yj}. 
		Then
		\begin{eqnarray*}
			&&\|\nabla_{\vz} \widetilde E(\vz)\|^2\\
			&\le&2\cdot\left[15\lambda(1+\delta)^2\|\vx\|^4\|\vh\|^2+\frac{27\lambda(1+\delta)\|\vx\|^2}{N}\sum^N_{j=1}|\vh^*A_j\vh|^2\right]+2\mu^2(1+\varepsilon_1)^2\left(2+\varepsilon_1+\delta\right)^2\|\vx\|^4\|\vh\|^2\\
			&\le&2\phi\lambda\|\vx\|^4\|\vh\|^2+\frac{54\lambda(1+\delta)\|\vx\|^2}{N}\sum^N_{j=1}|\vh^*A_j\vh|^2\\
			&\le&\widetilde L_{0,\delta}\frac{\lambda}{\zeta}\|\vx\|^2\Bigg(\frac{\zeta(1-\delta)}{8}\|\vx\|^2\dist^2(\vz, \vx)+\frac{\zeta}{10\lambda N}\sum_{j=1}^{N}\big|(\vz-\vx e^{i\theta(\vz)})^*A_j(\vz-\vx e^{i\theta(\vz)})\big|^2\Bigg),
		\end{eqnarray*} 
		by setting
		$\widetilde L_{0,\delta}=\max\{16\phi/(1-\delta),540(1+\delta)\}$ with $\phi=15(1+\delta)^2+(1+\varepsilon_1)^2(2+\varepsilon_1+\delta)^2$, where
		the first inequality follows from \eqref{eq:bound_e_tilde}, \eqref{eq:bound_regular}, and \eqref{R_cond}.  
	\end{proof}  
	
	%################################################################################
	\section{Proof of Theorems}\label{sec:proof_theorem}
	\subsection{Proof of Theorem \ref{theo:convergence}.}
	\label{subsec:proof-convergence}
	\begin{proof}
		In the following, we prove $\vz_{k+1} \in \calS(\vx,\varepsilon_0)$  and $$\dist(\vz_{k+1},\vx)\le\sqrt{1-\frac{\xi\zeta(1-\delta)}{4}\|\vx\|^2}\dist(\vz_k,\vx)$$
		if $\vz_k \in \calS(\vx,\varepsilon_0)$. According to Lemma \ref{lem:lem_regularity_cond}, for every $0<\delta<1$ and $\vz\in\calS(\vx,\varepsilon_0)$ we have  
		\begin{equation}\label{curvature_condition}
			\re\left(\langle\nabla_{\vz} E(\vz),\,\vz-\vx e^{i\theta(\vz)}\rangle\right)\ge \frac{\zeta(1-\delta)}{4}\|\vx\|^2\cdot\dist^2(\vz,\vx)+\frac{1}{10N}\sum_{j=1}^{N}\big|(\vz-\vx e^{i\theta(\vz)})^*A_j(\vz-\vx e^{i\theta(\vz)})\big|^2
		\end{equation}
		with probability at least $ 1-7\exp(- c_{2,\delta} d)-\exp(-\gamma_2 d\log^2d)-Nd^{-7r}$ when $ N $ satisfies \eqref{eq:N2}.
		
		According to Lemma \ref{lem:lem_smooth_cond}, 
		when $ N $ satisfies \eqref{eq:N2}, we have for all $\delta>0$ and $\vz\in\calS(\vx,\varepsilon_0)$
		\begin{equation}	\label{smooth_cond}
			\|\nabla_{\vz} E(\vz)\|^2
			\leq L_{0,\delta}\frac{\lambda}{\zeta}\|\vx\|^2\Bigg(\frac{\zeta(1-\delta)}{8}\|\vx\|^2\dist^2(\vz, \vx)+\frac{1}{10N}\sum_{j=1}^{N}\big|(\vz-\vx e^{i\theta(\vz)})^*A_j(\vz-\vx e^{i\theta(\vz)})\big|^2\Bigg)
		\end{equation}
		with probability at least $ 1-7\exp(-c_{3,\delta} d) -Nd^{-7r}$. 
		Fix $k\ge0$ and assume $\vz_k\in\calS(\vx,\varepsilon_0)$. Since the curvature bound contains twice the first term appearing in the smoothness bound, \eqref{curvature_condition} and \eqref{smooth_cond} give
		\[
		\re\big(\langle\nabla_{\vz} E(\vz_k),\,\vz_k-\vx e^{i\theta(\vz_k)}\rangle\big)\geq\frac{\zeta(1-\delta)}{8}\|\vx\|^2\|\vz_k-\vx e^{i\theta(\vz_k)}\|^2+\frac{\zeta}{L_{0,\delta}\lambda\|\vx\|^2}\|\nabla_{\vz} E(\vz_k)\|^2
		\]
		with probability greater than  $ 1-7\exp(-c_{0,\delta} d)-\exp(-\gamma_0d\log^2d)-Nd^{-7r}$ provided $ N $ satisfying \eqref{eq:N2}. 
		\noindent
		Then 
		\begin{align*}
			&\dist(\vz_{k+1},\vx)^2=\|\vz_{k+1}-\vx e^{i\theta(\vz_{k+1})}\|^2
			\le\|\vz_{k+1}-\vx e^{i\theta(\vz_k)}\|^2
			\le\|\vz_k-\vx e^{i\theta(\vz_k)}-\xi\cdot\nabla_{\vz} E(\vz_k)\|^2\\
			=&\|\vz_k-\vx e^{i\theta(\vz_k)}\|^2-2\xi\cdot\re\big(\langle\nabla_{\vz} E(\vz_k),\,\vz_k-\vx e^{i\theta(\vz_k)}\rangle\big)+\xi^2\|\nabla_{\vz} E(\vz_k)\|^2\\
			\le&\|\vz_k-\vx e^{i\theta(\vz_k)}\|^2-2\xi\cdot\Bigg(\frac{\zeta(1-\delta)}{8}\|\vx\|^2\|\vz_k-\vx e^{i\theta(\vz_k)}\|^2+\frac{\zeta}{L_{0,\delta}\lambda\|\vx\|^2}\|\nabla_{\vz} E(\vz_k)\|^2\Bigg)+\xi^2\|\nabla_{\vz}E(\vz_k)\|^2\\
			=&\left(1-\frac{\xi\zeta(1-\delta)}{4}\|\vx\|^2\right)\|\vz_k-\vx e^{i\theta(\vz_k)}\|^2+\xi\left(\xi-\frac{2\zeta}{L_{0,\delta}\lambda\|\vx\|^2}\right)\|\nabla_{\vz}E(\vz_k)\|^2\\
			\leq&\left(1-\frac{\xi\zeta(1-\delta)}{4}\|\vx\|^2\right)\|\vz_k-\vx e^{i\theta(\vz_k)}\|^2=\left(1-\frac{\xi\zeta(1-\delta)}{4}\|\vx\|^2\right)\dist(\vz_k,\vx)^2,
		\end{align*}
		Since  $0<\zeta\le\lambda\le1$, $L_{0,\delta}>1$ and $0<\xi\le\frac{2\zeta}{L_{0,\delta}\lambda\|\vx\|^2}$,
		\[
		0<\frac{\xi\zeta(1-\delta)}4\|\vx\|^2
		\le\frac{\zeta^2(1-\delta)}{2L_{0,\delta}\lambda}
		\le\frac1{2L_{0,\delta}}<1.
		\]
		Thus the contraction factor lies in $(0,1)$, and \eqref{eq:local_linear_rate} together with $\vz_{k+1}\in\calS(\vx,\varepsilon_0)$ follows. Induction from $\vz_0\in\calS(\vx,\varepsilon_0)$ completes the proof.
	\end{proof}
	\subsection{Proof of Theorem \ref{theo:relation}.}\label{subsec:proof_convergence1}
	\begin{proof}
		According to (2) in Lemma \ref{lem:lem_regularity_cond}, for all $\delta>0$ and $\dist(\vz,\vx)\in[\frac{\varepsilon_0}2\|\vx\|,\varepsilon_1\|\vx\|]$ we have  
		\begin{equation}\label{curvature_condition_tilde}
			\re\left(\langle\nabla_{\vz} \widetilde E(\vz),\vz-\vx e^{i\theta(\vz)}\rangle\right)\ge \frac{\zeta(1-\delta)}{4}\|\vx\|^2\cdot\dist^2(\vz,\vx)+\frac{\zeta}{10\lambda N}\sum_{j=1}^{N}\big|(\vz-\vx e^{i\theta(\vz)})^*A_j(\vz-\vx e^{i\theta(\vz)})\big|^2,
		\end{equation}
		with probability at least  $1-7\exp(-c_{2,\delta} d)-Nd^{-7r}
		-\exp(-\gamma_2 d\log^2 d)$ when $N$ satisfies \eqref{eq:N_tilde}, since 
		\[
		\frac dr\left(\frac\mu\zeta\right)^2\frac\lambda\zeta
		\le
		\frac{d^2}{d-r}\left(\frac\lambda\zeta\right)^2,
		\qquad
		d\left(\frac\lambda\zeta\right)^2
		\le
		\frac{d^2}{d-r}\left(\frac\lambda\zeta\right)^2.
		\]
		According to Lemma \ref{lem:lem_smooth_cond}, 
		when $N$ satisfies \eqref{eq:N_tilde}, we have for all $\delta>0$ and $\dist(\vz,\vx)\in[\frac{\varepsilon_0}2\|\vx\|,\varepsilon_1\|\vx\|]$
		\begin{equation}	\label{smooth_cond_tilde}
			\|\nabla_{\vz} \widetilde E(\vz)\|^2
			\leq \widetilde L_{0,\delta}\frac{\lambda}{\zeta}\|\vx\|^2\Bigg(\frac{\zeta(1-\delta)}{8}\|\vx\|^2\dist^2(\vz, \vx)+\frac{\zeta}{10\lambda N}\sum_{j=1}^{N}\big|(\vz-\vx e^{i\theta(\vz)})^*A_j(\vz-\vx e^{i\theta(\vz)})\big|^2\Bigg)
		\end{equation}
		with probability at least $1-7\exp(-c_{3,\delta} d)-Nd^{-7r}$, where $\widetilde L_{0,\delta}>0$ depends on $\delta$.
		
		\noindent
		Combining \eqref{curvature_condition_tilde} and \eqref{smooth_cond_tilde} yields
		\[
		\re\big(\langle\nabla_{\vz} \widetilde E(\vz),\,\vz-\vx e^{i\theta(\vz)}\rangle\big)\geq\frac{\zeta(1-\delta)}{8}\|\vx\|^2\|\vz-\vx e^{i\theta(\vz)}\|^2+\frac{\zeta}{\widetilde L_{0,\delta}\lambda\|\vx\|^2}\|\nabla_{\vz} \widetilde E(\vz)\|^2
		\]
		with probability greater than  $1-7\exp(-\widetilde c_{0,\delta} d)- Nd^{-7r}
		-\exp(-\widetilde \gamma_0 d\log^2 d)$ provided $N$ satisfying \eqref{eq:N_tilde}.  
		\noindent
		Then 
		\begin{align*}
			&\dist(\widetilde\vz,\vx)^2=\|\widetilde\vz-\vx e^{i\theta(\widetilde\vz)}\|^2
			\le\|\widetilde\vz-\vx e^{i\theta(\vz)}\|^2
			\le\|\vz-\vx e^{i\theta(\vz)}-\widetilde\xi\cdot\nabla_{\vz} \widetilde E(\vz)\|^2\\
			=&\|\vz-\vx e^{i\theta(\vz)}\|^2-2\widetilde\xi\cdot\re\big(\langle\nabla_{\vz} \widetilde E(\vz),\,\vz-\vx e^{i\theta(\vz)}\rangle\big)+\widetilde\xi^2\|\nabla_{\vz} \widetilde E(\vz)\|^2\\
			\le&\|\vz-\vx e^{i\theta(\vz)}\|^2-2\widetilde\xi\cdot\Bigg(\frac{\zeta(1-\delta)}{8}\|\vx\|^2\|\vz-\vx e^{i\theta(\vz)}\|^2+\frac{\zeta}{\widetilde L_{0,\delta}\lambda\|\vx\|^2}\|\nabla_{\vz} \widetilde E(\vz)\|^2\Bigg)+\widetilde\xi^2\|\nabla_{\vz}\widetilde E(\vz)\|^2\\
			=&\left(1-\frac{\widetilde\xi\zeta(1-\delta)}{4}\|\vx\|^2\right)\|\vz-\vx e^{i\theta(\vz)}\|^2+\widetilde\xi\left(\widetilde\xi-\frac{2\zeta}{\widetilde L_{0,\delta}\lambda\|\vx\|^2}\right)\|\nabla_{\vz}\widetilde E(\vz)\|^2\\
			\leq&\left(1-\frac{\widetilde\xi\zeta(1-\delta)}{4}\|\vx\|^2\right)\|\vz-\vx e^{i\theta(\vz)}\|^2=\left(1-\frac{\widetilde\xi\zeta(1-\delta)}{4}\|\vx\|^2\right)\dist(\vz,\vx)^2,
		\end{align*}
		Since $0<\zeta\le\lambda\le1$, $\widetilde L_{0,\delta}\ge1$ and $0<\widetilde\xi\le \frac{2\zeta}{\widetilde L_{0,\delta}\lambda\|\vx\|^2}$, also gives
		\[
		0<\frac{\widetilde\xi\zeta(1-\delta)}4\|\vx\|^2
		\le\frac{2\zeta^2(1-\delta)}{4\widetilde L_{0,\delta}\lambda}
		\le\frac1{2}.
		\]
		Thus the contraction factor belongs to $[1/2,1)$.
		So we have $\dist^2(\widetilde\vz,\vx)\le\left(1-\frac{\widetilde\xi\zeta(1-\delta)}{4}\|\vx\|^2\right)\dist^2(\vz,\vx)$ when
		$\dist(\vz,\vx)\in[\frac{\varepsilon_0}2\|\vx\|,\varepsilon_1\|\vx\|]$.
	\end{proof}
	\subsection{Proof of Theorem \ref{theo:regularized_inner_invariance}}\label{subsec:proof-invariance}
	\begin{proof}
		Put
		\[
		\vh=e^{-i\theta(\vz)}\vz-\vx,\qquad
		s=\|\vh\|=\dist(\vz,\vx),\qquad
		\vz^+=\vz-\widetilde\xi\nabla_{\vz}\widetilde E(\vz).
		\]
		Next, we consider two cases.
		
		\noindent
		\textbf{Case 1: $0\le s\le\varepsilon_0\|\vx\|/2$.} Since $0\preceq A_j\preceq I_d$,
		\begin{equation}\label{eq:bd1}
			(\vh^*A_j\vh)^2
			\le s^2\vh^*A_j\vh.
		\end{equation}
		Theorem \ref{theo:concentration1} together with \eqref{eq:bd1} implies
		\[
		\frac1N\sum_{j=1}^N(\vh^*A_j\vh)^2
		\le
		s^2\vh^*
		\left(\frac1N\sum_{j=1}^NA_j\right)\vh
		\le
		(1+\delta)\frac rd\,s^4.
		\]
		Since
		\[
		s\le\frac{\varepsilon_0}{2}\|\vx\|,
		\qquad
		\varepsilon_0^2=\frac{10\zeta}{27\lambda},
		\qquad
		\lambda\le1,\qquad \frac rd\le1,
		\]
		we obtain by Part~(1) of Lemma~\ref{lem:lem_smooth_cond}
		\begin{align*}
			\|\nabla_{\vz} E(\vz)\|^2
			&\le
			L_{0,\delta}\frac{\lambda}{\zeta}\|\vx\|^2
			\left(
			\frac{\zeta(1-\delta)}8\|\vx\|^2s^2
			+\frac{1+\delta}{10}\frac rd\,s^4
			\right)\le
			L_{0,\delta}
			\left(
			\frac{1-\delta}{32}
			+\frac{1+\delta}{432}
			\right)
			\varepsilon_0^2\|\vx\|^6.
		\end{align*}
		Theorem \ref{theo:concentration1} implies
		\[
		|\rho^2-\|\vx\|^2|
		\le
		\frac{\delta\zeta\varepsilon_0}
		{8(1+\varepsilon_1)\mu}\|\vx\|^2
		\le
		\frac{\delta\varepsilon_0}{2}\|\vx\|^2,
		\]
		with probability at least $1-\exp(-\widetilde c_{1,\delta} d)$, where we used $\zeta\le\mu$. Moreover,
		\[
		\big|\|\vz\|^2-\|\vx\|^2\big|
		\le s(2\|\vx\|+s),
		\qquad
		\|\vz\|\le\|\vx\|+s.
		\]
		Since $s\le\varepsilon_0\|\vx\|/2$,
		$\varepsilon_0\le\varepsilon_1$, and $\mu\le1$, it follows that
		\[
		\mu\big|\|\vz\|^2-\rho^2\big|\,\|\vz\|
		\le
		\frac12
		\left(2+\frac{\varepsilon_1}{2}+\delta\right)
		\left(1+\frac{\varepsilon_1}{2}\right)
		\varepsilon_0\|\vx\|^3.
		\]
		Thus, from
		\[
		\nabla_{\vz}\widetilde E(\vz)
		=
		\nabla_{\vz} E(\vz)-\mu(\|\vz\|^2-\rho^2)\vz,
		\]
		we have
		\[
		\|\nabla_{\vz}\widetilde E(\vz)\|
		\le
		C_{\mathrm{in},\delta}\varepsilon_0\|\vx\|^3,
		\]
		where
		\[
		C_{\mathrm{in},\delta}
		:=
		\left[
		L_{0,\delta}
		\left(
		\frac{1-\delta}{32}
		+\frac{1+\delta}{432}
		\right)
		\right]^{1/2}
		+
		\frac12
		\left(2+\frac{\varepsilon_1}{2}+\delta\right)
		\left(1+\frac{\varepsilon_1}{2}\right).
		\]
		Since $L_{0,\delta}$ depends only on $\delta$ and
		$\varepsilon_1=\sqrt{10/27}$ is fixed,
		$C_{\mathrm{in},\delta}$ also depends only on $\delta$. Since $$L_{1,\delta}=\max\left\{\frac{16\phi}{1-\delta},540(1+\delta),4\left[
		L_{0,\delta}
		\left(
		\frac{1-\delta}{32}
		+\frac{1+\delta}{432}
		\right)
		\right]^{1/2}
		+
		2\left(2+\frac{\varepsilon_1}{2}+\delta\right)
		(1+\frac{\varepsilon_1}{2}\right\}
		\geq4C_{\mathrm{in},\delta},$$
		with $\phi=15(1+\delta)^2+(1+\varepsilon_1)^2(2+\varepsilon_1+\delta)^2$, the step-size condition and $\zeta\le\lambda$ then give
		\[
		\widetilde\xi\|\nabla_{\vz}\widetilde E(\vz)\|
		\le
		\frac{2\zeta C_{\mathrm{in},\delta}}
		{L_{1,\delta}\lambda}
		\varepsilon_0\|\vx\|
		\le
		\frac{\varepsilon_0}{2}\|\vx\|.
		\]
		Consequently,
		\[
		\dist(\vz^+,\vx)
		\le
		s+\widetilde\xi\|\nabla_{\vz}\widetilde E(\vz)\|
		\le
		\varepsilon_0\|\vx\|.
		\]
		
		\noindent
		\textbf{Case 2:
			$\varepsilon_0\|\vx\|/2\le s\le\varepsilon_0\|\vx\|$.} 
		Theorem \ref{theo:relation} implies
		\begin{align*}
			\dist^2(\vz^+,\vx)
			&\le
			\|\vh-\widetilde\xi e^{-i\theta(\vz)}\nabla_{\vz}\widetilde E(\vz)\|^2\le
			\left(
			1-\frac{\widetilde\xi\zeta(1-\delta)}4\|\vx\|^2
			\right)s^2
			\le s^2
			\le\varepsilon_0^2\|\vx\|^2
		\end{align*}
		with probability at least $1-7\exp(-c_{1,\delta} d)-Nd^{-7r}
		-\exp(-\gamma_1 d\log^2 d)$, where the last inequality follows from
		\[
		0<\widetilde\xi
		\le
		\frac{2\zeta}
		{ L_{1,\delta}\lambda\|\vx\|^2},\  L_{1,\delta}\geq\max\left\{\frac{16\phi}{1-\delta},540(1+\delta)\right\}.
		\] 
		
		\noindent
		Hence $\vz^+\in\calS(\vx,\varepsilon_0)$ in both cases, which proves
		the invariance of the regularized inner ball.
	\end{proof}
	
	\subsection{Proof of Theorem \ref{theo:good_initialization_2stage}}\label{subsec:proof_2stage_spectral_init}
	\begin{proof}
		Let $\hat{\vz}_0$ be the eigenvector of $Y$ corresponding to the largest eigenvalue $\omega$ with $\|\hat{\vz}_0\|=1$. By Lemma \ref{lem:expectations},  $\E(Y)=\mu\|\vx\|^2I_d+(\lambda-\mu)\vx\vx^*=\mu\|\vx\|^2I_d+\zeta\vx\vx^*$. Note that the largest eigenvalue of $\E(Y)$ is $\lambda\|\vx\|^2$, we have by Weyl's Inequality about perturbation
		\begin{equation}
			\label{eq:theo_init1_resampling}
			|\omega-\lambda\|\vx\|^2|\le \|Y-(\mu\|\vx\|^2I_d+\zeta\vx\vx^*)\|. 
		\end{equation}
		Note that
		\begin{equation}
			\label{eq:theo_init2_resampling}
			\begin{split}
				\|Y-\mu\|\vx\|^2I_d-\zeta\vx\vx^*\| \ \ge&|\hat{\vz}_0^*(Y-\mu\|\vx\|^2I_d-\zeta\vx\vx^*)\hat{\vz}_0|\\
				=&|\omega-\mu\|\vx\|^2-\zeta|\hat{\vz}_0 ^*\vx|^2|\ge\zeta\|\vx\|^2\left(1-\frac{|\hat{\vz}_0^*\vx|^2}{\|\vx\|^2}\right)-|\omega-\lambda\|\vx\|^2|.
			\end{split}
		\end{equation}
		We first focus on the spectral part of the initializer $\widetilde\vz_0$. Set
		\[
		\eta:=\frac{\varepsilon_1^2}{5}.
		\]
		Apply Theorem~\ref{theo:concentration1-1} with its accuracy parameter equal to $4\eta$, 
		when $N\ge C_{1,\delta} d(\lambda/\zeta)^2\log^2 d$,
		\begin{equation}
			\label{eq:theo_init3_resampling}
			\|Y-\mu\|\vx\|^2I_d-\zeta\vx\vx^*\|\le \eta\zeta\|\vx\|^2.
		\end{equation}
		Combining \eqref{eq:theo_init1_resampling}, \eqref{eq:theo_init2_resampling} and \eqref{eq:theo_init3_resampling} we have
		\[\frac{|\hat{\vz}_0^*\vx|^2}{\|\vx\|^2}\ge1-2\eta.\]
		Since $\theta(\vz)=\mathop{\arg\min}_{\theta\in[0,2\pi)}\|\vz-\vx e^{i\theta}\|$ for all $\vz$. 
		Note that $\widetilde \vz_0=\rho\hat \vz_0$. 
		By the definition of the minimizing phase,
		$\re(e^{i\theta(\hat\vz_0)}\hat\vz_0^*\vx)=|\hat\vz_0^*\vx|$.  Hence
		$|\hat\vz_0^*\vx|/\|\vx\|\ge\sqrt{1-2\eta}$. 
		Applying the scalar consequence of Theorem~\ref{theo:concentration1} with relative tolerance $\eta$ also gives
		\[
		(1-\eta)\|\vx\|^2\leq \rho^2\leq (1+\eta)\|\vx\|^2.
		\]
		Combining this norm estimate with the direction estimate yields
		$\rho|\hat{\vz}_0^*\vx|/\|\vx\|\ge\sqrt{1-2\eta}\sqrt{1-\eta}\,\|\vx\|$. 
		Then 
		\begin{equation*}
			\begin{split}
				\dist^2(\widetilde\vz_0,\vx)=&\rho^2+\|\vx\|^2-2\rho
				|\hat\vz_0^*\vx|
				\le(1+\eta)\|\vx\|^2+\|\vx\|^2-2\sqrt{1-\eta}\|\vx\|^2\sqrt{1-2\eta}
				\le5\eta\|\vx\|^2=\varepsilon_1^2\|\vx\|^2,
			\end{split}
		\end{equation*}
		when $N\ge C_{1,\delta}\frac dr$ by Theorem \ref{theo:concentration1}. Thus $\widetilde\vz_0\in\calS(\vx,\varepsilon_1)$.
		
		We next prove that the prescribed number of regularized steps is sufficient.  This part is deterministic once the preceding concentration events hold.  Put
		\[
		q:=1-\frac{\widetilde\xi\zeta(1-\delta)}4\|\vx\|^2\in[1/2,1)
		\quad\text{and}\quad
		\tau:=\inf\{b\ge0:\dist(\widetilde\vz_b,\vx)\le\varepsilon_0\|\vx\|\}.
		\] 
		For every $b<\tau$, the definition of $\tau$ gives
		$\dist(\widetilde\vz_b,\vx)>\varepsilon_0\|\vx\|$.  On the other hand, induction using Theorem~\ref{theo:relation} gives
		\[
		\dist^2(\widetilde\vz_b,\vx)
		\le q^b\dist^2(\widetilde\vz_0,\vx)
		\le q^b\varepsilon_1^2\|\vx\|^2
		\le\varepsilon_1^2\|\vx\|^2.
		\]
		Thus every iterate before $\tau$ remains in the annulus where Theorem~\ref{theo:relation} applies.
		If $\tau>B:=B(\widetilde\xi)=\left\lceil\frac{4}{\widetilde\xi\zeta(1-\delta)\|\vx\|^2}\log\!\left(\frac{\lambda}{\zeta}\right)\right\rceil$, then, using $q^B\le e^{-B\widetilde\xi\zeta(1-\delta)\|\vx\|^2/4}\le\frac{\zeta}{\lambda}$ and $\varepsilon_1^2/\varepsilon_0^2=\lambda/\zeta$, we obtain
		\[
		\dist^2(\widetilde\vz_B,\vx)\le q^B\varepsilon_1^2\|\vx\|^2\le\varepsilon_0^2\|\vx\|^2,
		\]
		contradicting $\tau>B$. Therefore $\tau\le B(\widetilde\xi)$. Theorem~\ref{theo:regularized_inner_invariance} then gives $\widetilde\vz_b\in\calS(\vx,\varepsilon_0)$ for every $b\ge\tau$, and in particular $\widetilde\vz_{B(\widetilde\xi)}\in\calS(\vx,\varepsilon_0)$.
	\end{proof}
	%########################################################################
	\section{Concentration Inequalities}
	\label{sec:concentration}
	In this section, our emphasis is on the concentration inequalities utilized in the proofs. Specifically, we will begin by presenting certain expectations associated with random orthogonal projection matrices.
	\begin{lemma}[Scalar distribution and moments]
		\label{lem:beta-d}
		Let $A$ be an orthogonal projection in $\C^d$ distributed from the Haar measure on the set of rank-$r$ orthogonal projection matrices, where $1\le r<d$. Then, for every fixed nonzero vector $\vx\in\C^d$,
		\[
		\frac{\vx^*A\vx}{\|\vx\|^2}\sim\operatorname{Beta}(r,d-r).
		\]
		Consequently, for every integer $p\ge 1$,
		\[
		\E\left[(\vx^*A\vx)^p\right]=\|\vx\|^{2p}\frac{(r)_p}{(d)_p}\leq\|\vx\|^{2p}\Gamma(p+1)\left(\frac{r}{d}\right)^p,
		\]
		where
		\[
		(r)_p=r(r+1)\cdots(r+p-1),\qquad (d)_p=d(d+1)\cdots(d+p-1).
		\]
	\end{lemma}
	
	\begin{proof}
		Let $U$ be Haar-distributed on the unitary group $\mathcal U(d)$ and let $\Sigma_r=\operatorname{diag}(I_r,0_{d-r})$. Then
		\[
		A\stackrel{d}{=}U\Sigma_rU^*,
		\]
		where $\stackrel{d}{=}$ represents an identity in distribution. 
		Let $\vu=\vx/\|\vx\|$. Then
		\[
		\frac{\vx^*A\vx}{\|\vx\|^2}=\vu^*A\vu\stackrel{d}{=}\vu^*U\Sigma_rU^*\vu=\vb^*\Sigma_r\vb=\sum_{k=1}^r|b_k|^2,
		\]
		where $\vb=U^*\vu=[b_1,\ldots,b_d]^\top$ is uniformly distributed on the unit sphere of $\C^d$. Thus, $\vb$ has the same distribution as $\vg/\|\vg\|$, where $\vg=[g_1,\ldots,g_d]^\top$ has independent standard complex Gaussian entries. Hence
		\[
		\sum_{k=1}^r|b_k|^2\stackrel{d}{=}\frac{\sum_{k=1}^r|g_k|^2}{\sum_{k=1}^d|g_k|^2}=\frac{X_1}{X_1+X_2},
		\]
		where $X_1=\sum_{k=1}^r|g_k|^2\sim\Gamma(r,1)$ and $X_2=\sum_{k=r+1}^d|g_k|^2\sim\Gamma(d-r,1)$ are independent. Therefore,
		\[
		\frac{\vx^*A\vx}{\|\vx\|^2}=\vu^*A\vu\sim\operatorname{Beta}(r,d-r).
		\]
		
		Define $Y=\vu^*A\vu\sim\operatorname{Beta}(r,d-r)$. For every integer $p\ge1$, by the beta integral,
		\[
		\E Y^p=\frac{\Gamma(d)}{\Gamma(r)\Gamma(d-r)}\int_0^1t^{p+r-1}(1-t)^{d-r-1}\,dt=\frac{\Gamma(r+p)\Gamma(d)}{\Gamma(r)\Gamma(d+p)}=\frac{(r)_p}{(d)_p}.
		\]
		Therefore,
		\[
		\E\big[(\vx^*A\vx)^p\big]=\|\vx\|^{2p}\E\big[(\vu^*A\vu)^p\big]=\|\vx\|^{2p}\frac{(r)_p}{(d)_p}.
		\]
		Finally, since $r+j\le r(j+1)$ and $d+j\ge d$ for $j=0,\ldots,p-1$,
		\[
		\frac{(r)_p}{(d)_p}=\prod_{j=0}^{p-1}\frac{r+j}{d+j}\le\frac{r^pp!}{d^p}=\Gamma(p+1)\left(\frac{r}{d}\right)^p.
		\]
		This completes the proof.
	\end{proof}
	
	% ------------------ ===============================================
	\begin{lemma}
		\label{lem:expectation1}
		Let $A$ be a rank-$r$ orthogonal projection in $\C^d$ distributed from the Haar measure, where $1\le r<d$. Then:
		\begin{enumerate}
			\item For any $i\ne j$,
			\[
			\E(a_{ij}^2)=\E(a_{12}^2),\qquad \E(|a_{ij}|^2)=\E(|a_{12}|^2).
			\]
			\item For every unitary matrix $P\in\C^{d\times d}$ and every $\vx\in\C^d$, define
			$
			\mathcal E(\vx)=\E_A\big[(\vx^*A\vx)A\big]$. 
			If $\widetilde{\vx}=P\vx$, then
			\[
			\mathcal E(\vx)=P^*\mathcal E(\widetilde{\vx})P.
			\]
		\end{enumerate}
	\end{lemma}
	
	\begin{proof}
		For any unitary matrix $P$, the random matrices $A$ and $PAP^*$ have the same distribution. Given $i\ne j$, choose a permutation matrix $P$ whose corresponding permutation maps $1$ to $i$ and $2$ to $j$. Then the $(1,2)$-th entry of $P^*AP$ is $a_{ij}$, and hence
		\[
		\E(a_{ij}^2)=\E(a_{12}^2),\qquad \E(|a_{ij}|^2)=\E(|a_{12}|^2).
		\]
		
		For the second assertion, let $B=PAP^*$. Since $\widetilde{\vx}=P\vx$, we have
		\[
		\vx^*A\vx=\widetilde{\vx}^*B\widetilde{\vx},\qquad A=P^*BP.
		\]
		Therefore,
		\[
		\mathcal E(\vx)=P^*\E_B\big[(\widetilde{\vx}^*B\widetilde{\vx})B\big]P=P^*\mathcal E(\widetilde{\vx})P.
		\]
	\end{proof}
	
	%######################################################################
	
	\begin{lemma}
		\label{lem:expectations}
		Let $A$ be a rank-$r$ orthogonal projection in $\C^d$ distributed from the Haar measure, where $1\le r<d$. Then
		\[
		\E A=\frac rd I_d.
		\]
		Moreover, for every fixed $\vx\in\C^d$,
		\[
		\E\big[(\vx^*A\vx)A\big]=\mu\|\vx\|^2I_d+\zeta\vx\vx^*,\qquad\E\big[A\vx(A\vx)^*\big]=\zeta\|\vx\|^2I_d+\mu\vx\vx^*,\qquad\E\big[A\vx(A\vx)^\top\big]=\lambda\vx\vx^\top,
		\]
		where
		\begin{equation}\label{defi_para}
			\lambda=\frac{r(r+1)}{d(d+1)},\qquad
			\mu=\frac{dr^2-r}{d(d+1)(d-1)},\qquad
			\zeta=\lambda-\mu=\frac{r(d-r)}{d(d+1)(d-1)}.
		\end{equation}
	\end{lemma}
	
	\begin{proof}
		For every unitary matrix $P$, the random matrices $A$ and $PAP^*$ have the same distribution. Hence
		\[
		P(\E A)P^*=\E A.
		\]
		Thus $\E A$ commutes with every unitary matrix and must be a scalar multiple of the identity. Since $\tr(A)=r$,
		\[
		\E A=\frac rd I_d.
		\]
		
		We first consider $\vx=\ve_1$. Set
		\[
		M_1=\E\big[(\ve_1^*A\ve_1)A\big].
		\]
		For every unitary matrix $Q\in\C^{(d-1)\times(d-1)}$, let $P=\operatorname{diag}(1,Q)$. Since $P\ve_1=\ve_1$, Lemma~\ref{lem:expectation1} gives $M_1=P^*M_1P$. It follows that
		\[
		M_1=\begin{bmatrix}\lambda&0\\0&\mu I_{d-1}\end{bmatrix}.
		\]
		By Lemma~\ref{lem:beta-d},
		\[
		\lambda=\E(a_{11}^2)=\frac{r(r+1)}{d(d+1)}.
		\]
		Moreover,
		\[
		\tr(M_1)=\E\big[a_{11}\tr(A)\big]=r\E(a_{11})=\frac{r^2}{d}.
		\]
		Therefore,
		\[
		\lambda+(d-1)\mu=\frac{r^2}{d},
		\]
		which yields
		\[
		\mu=\frac{dr^2-r}{d(d+1)(d-1)}.
		\]
		Since $\zeta=\lambda-\mu$, we obtain
		\[
		M_1=\mu I_d+\zeta \ve_1\ve_1^*.
		\]
		
		For a general nonzero vector $\vx$, choose a unitary matrix $P$ such that $P\vx=\|\vx\|\ve_1$. By Lemma~\ref{lem:expectation1} and the homogeneity of $\mathcal E$,
		\[
		\E\big[(\vx^*A\vx)A\big]=P^*\mathcal E(\|\vx\|\ve_1)P
		=\mu\|\vx\|^2I_d+\zeta\vx\vx^*.
		\]
		The identity is trivial for $\vx=0$.
		
		Next, set
		\[
		M_2=\E\big[A\ve_1(A\ve_1)^*\big].
		\]
		The same invariance argument gives
		\[
		M_2=\begin{bmatrix}\lambda&0\\0&\zeta I_{d-1}\end{bmatrix}.
		\]
		Indeed, its $(1,1)$-th entry is $\E(a_{11}^2)=\lambda$, while
		\[
		\tr(M_2)=\E\|A\ve_1\|^2=\E(\ve_1^*A^2\ve_1)=\E(a_{11})=\frac rd.
		\]
		Thus
		\[
		\lambda+(d-1)\zeta=\frac rd,
		\]
		and consequently
		\[
		\zeta=\frac{r(d-r)}{d(d+1)(d-1)}.
		\]
		Since $\lambda-\zeta=\mu$, we have
		\[
		M_2=\zeta I_d+\mu \ve_1\ve_1^*.
		\]
		Applying unitary invariance as above yields
		\[
		\E\big[A\vx(A\vx)^*\big]=\zeta\|\vx\|^2I_d+\mu\vx\vx^*.
		\]
		
		Finally, let
		\[
		M_3=\E\big[A\ve_1(A\ve_1)^\top\big].
		\]
		For a diagonal unitary matrix $D=\operatorname{diag}(1,e^{i\theta_2},\ldots,e^{i\theta_d})$, the matrices $A$ and $DAD^*$ have the same distribution. Hence
		\[
		M_3=DM_3D^\top.
		\]
		Since the phases $\theta_2,\ldots,\theta_d$ are arbitrary, all entries of $M_3$ vanish except possibly its $(1,1)$-th entry. As
		\[
		(M_3)_{11}=\E(a_{11}^2)=\lambda,
		\]
		we obtain
		\[
		M_3=\lambda \ve_1\ve_1^\top.
		\]
		A final application of unitary invariance gives
		\[
		\E\big[A\vx(A\vx)^\top\big]=\lambda\vx\vx^\top.
		\]
	\end{proof}
	
	% ---------------------------------------------------------------------------------	
	\begin{lemma}[Bernstein-type tail bound for beta variables]
		\label{lem:beta-concentration}
		Let $Y\sim\operatorname{Beta}(r,d-r)$, where $1\le r<d$. Then for every $t>0$,
		\[
		\Prob\left(\left|Y-\frac rd\right|\ge t\right)\le 2\exp\left[-\frac{1}{8}\min\left\{\frac{d^2t^2}{r},dt\right\}\right].
		\]
	\end{lemma}
	
	\begin{proof}
		Let $a=r/d$. By the gamma representation of the beta distribution, we may write
		\[
		Y=\frac{U}{U+V},
		\]
		where $U\sim\Gamma(r,1)$ and $V\sim\Gamma(d-r,1)$ are independent.
		
		We first consider the upper tail. Let $0<t<1-a$ and set $s=a+t$. For every $\theta>0$ satisfying $\theta(1-s)<1$, Chernoff's inequality gives
		\[
		\Prob(Y\ge s)=\Prob\big((1-s)U-sV\ge0\big)
		\le (1-\theta(1-s))^{-r}(1+\theta s)^{-(d-r)}.
		\]
		Choosing $\theta=t/[s(1-s)]$, we obtain
		\[
		\Prob(Y\ge a+t)\le \left(\frac{s}{a}\right)^r\left(\frac{1-s}{1-a}\right)^{d-r}
		=\exp\big(-dH_+(t)\big),
		\]
		where
		\[
		H_+(t)=-a\log\left(1+\frac ta\right)-(1-a)\log\left(1-\frac{t}{1-a}\right).
		\]
		Since
		\[
		H_+'(t)=\frac{t}{(a+t)(1-a-t)},
		\]
		and $1-a-t\le1$, we have
		\[
		H_+(t)\ge\int_0^t\frac{x}{a+x}\,dx.
		\]
		If $0<t\le a$, then
		\[
		H_+(t)\ge\int_0^t\frac{x}{2a}\,dx=\frac{t^2}{4a}.
		\]
		If $t>a$, then
		\[
		H_+(t)\ge\int_{t/2}^t\frac{x}{a+x}\,dx\ge\int_{t/2}^t\frac14\,dx=\frac t8,
		\]
		because $x\ge t/2$ and $a<t$ imply $x/(a+x)\ge1/4$. Therefore,
		\[
		H_+(t)\ge\frac18\min\left\{\frac{t^2}{a},t\right\},
		\]
		and hence
		\[
		\Prob(Y-a\ge t)\le\exp\left[-\frac18\min\left\{\frac{d^2t^2}{r},dt\right\}\right].
		\]
		For $t\ge1-a$, the upper-tail probability is zero, so the same estimate holds for every $t>0$.
		
		The lower tail is obtained similarly. Indeed, for $0<t<a$, setting $s=a-t$ and applying Chernoff's inequality to
		\[
		\{Y\le s\}=\{sV-(1-s)U\ge0\}
		\]
		with parameter $\theta=t/[s(1-s)]$ gives
		\[
		\Prob(Y-a\le-t)\le\exp\left\{d\left[(1-a)\log\left(1+\frac{t}{1-a}\right)+a\log\left(1-\frac ta\right)\right]\right\}.
		\]
		Using $\log(1+u)\le u$ for $u\ge0$ and $\log(1-v)\le-v-v^2/2$ for $0\le v<1$, we obtain
		\[
		\Prob(Y-a\le-t)\le\exp\left(-\frac{d^2t^2}{2r}\right)
		\le\exp\left[-\frac12\min\left\{\frac{d^2t^2}{r},dt\right\}\right].
		\]
		For $t\ge a$, the lower-tail probability is zero.
		In sum, we conclude the result.
	\end{proof}
	
	%########################################################################   
	Now we are ready for concentration inequalities.
	\begin{theorem}
		\label{theo:concentration1}
		Let $A_1,\ldots,A_N$ be i.i.d. rank-$r$ orthogonal projections in $\C^d$ drawn from the Haar measure, where $1\le r<d$. 
		There exist absolute constants $c_4,C_4>0$ such that, for every $0<\rho\le1$, when $N\ge C_4d/(r\rho^2)$,
		\[
		\Big\|\frac1N\sum_{j=1}^N A_j-\frac rd I_d\Big\|\le \rho\frac rd
		\]
		with probability at least $1-\exp(-c_4d)$. Let $\lambda,\mu,\zeta$ be defined in \eqref{defi_para}. Then, 
		for every $0<\delta\le1$, there exist constants $c_{5},C_{5,\delta}>0$ such that
		\begin{equation}\label{eq:N_sqrt_lambda_zeta}
			N\ge C_{5,\delta}\frac dr\left(\frac{\mu}{\zeta}\right)^2\frac{\lambda}{\zeta}
		\end{equation}
		implies
		
        \[
        \Big\|\frac1N\sum_{j=1}^N A_j-\frac rd I_d\Big\|\le
        \frac{\delta r\zeta}{8(1+\varepsilon_1)d\mu}\sqrt{\frac{10\zeta}{27\lambda}}.
        \]
		with probability at least $1-\exp(-c_{5}d)$. Moreover, for every $0<\delta\le1$, there exist constants $c_{6},C_{6,\delta}>0$ such that
		\begin{equation}\label{eq:N_constant}
			N\ge C_{6,\delta}\frac dr
		\end{equation}
		implies
		\[
		\Big\|\frac1N\sum_{j=1}^N A_j-\frac rd I_d\Big\|\le \delta\frac rd
		\]
		with probability at least $1-\exp(-c_{6}d)$.
	\end{theorem}
	\begin{proof}
		For a fixed unit vector $\vu\in\C^d$, define $Y_j=\vu^*A_j\vu$. By Lemma~\ref{lem:beta-d},
		\[
		Y_j\sim\operatorname{Beta}(r,d-r),\qquad \E Y_j=\frac rd.
		\]
		Set $H=\frac1N\sum_{j=1}^N A_j-\frac rd I_d$. 
		Lemma~\ref{lem:beta-concentration}  yields an absolute constant $c_{\mathrm B}>0$ such that, for every $t>0$,
		\[
		\Prob\left(|\vu^*H\vu|\ge t\right)
		\le 2\exp\left[-c_{\mathrm B}N\min\left\{\frac{d^2t^2}{r},dt\right\}\right].
		\]
		Taking $t=\rho r/(2d)$ and using $0<\rho\le1$, we obtain
		\[
		\Prob\left(|\vu^*H\vu|\ge\frac{\rho r}{2d}\right)\le 2\exp(-c_{\mathrm B}Nr\rho^2).
		\]
		Let $\mathcal N$ be a $1/4$-net of the unit sphere in $\C^d$ with $|\mathcal N|\le9^{2d}$. By the union bound,
		\[
		\Prob\left(\sup_{\vu\in\mathcal N}|\vu^*H\vu|\ge\frac{\rho r}{2d}\right)
		\le 2\exp\big(2d\log9-c_{\mathrm B}Nr\rho^2\big).
		\]
		Hence, if $N\ge C_4d/(r\rho^2)$ for a sufficiently large absolute constant $C_4$, then
		\[
		\sup_{\vu\in\mathcal N}|\vu^*H\vu|\le\frac{\rho r}{2d}
		\]
		with probability at least $1-\exp(-c_4d)$ for some absolute constant $c_4>0$. Since $H$ is Hermitian, the standard net argument gives
		\[
		\|H\|\le2\sup_{\vu\in\mathcal N}|\vu^*H\vu|\le\rho\frac rd.
		\]
		This proves the first assertion.
		
		For the second assertion, set
		$\rho=\frac\delta{8(1+\varepsilon_1)}\frac{\zeta}{\mu}\sqrt{\frac{10\zeta}{27\lambda}}$. 
		Since $\zeta\le\mu$ and $\zeta\le\lambda$, we have $0<\rho\le1$. Moreover, $
		\rho\frac rd=\frac\delta{8(1+\varepsilon_1)} \frac{r\zeta}{d\mu}\sqrt{\frac{10\zeta}{27\lambda}}$, 
		while the sample size condition in the first assertion becomes $
		N\ge {\frac{64\times 27}{10}(1+\varepsilon_1)^2}C_4\delta^{-2}\frac dr\left(\frac{\mu}{\zeta}\right)^2\frac{\lambda}{\zeta}$. 
		The claimed result follows by taking 
		$C_{5,\delta}={\frac{64\times 27}{10}(1+\varepsilon_1)^2}C_4\delta^{-2},\ c_{5}=c_4$. Finally, taking $\rho=\delta$ in the first assertion gives the last claim with $C_{6,\delta}=C_4\delta^{-2},\ c_{6}=c_4$.
	\end{proof}
	
	%########################################################################
	\begin{theorem}
		\label{theo:concentration1-1}
		Let $A_1,\ldots,A_N$ be i.i.d. rank-$r$ orthogonal projections in $\C^d$ drawn from the Haar measure, where $1\le r<d$, and let
		$\lambda$ and $\zeta$ be defined in \eqref{defi_para}. For every fixed
		$\vx\in\C^d$ and every $0<\delta<1$, there exist constants
		$c_{7,\delta},C_{7,\delta}>0$ such that
		\begin{equation}\label{N}
			N\ge C_{7,\delta} d
			\left(\frac{\lambda}{\zeta}\right)^2\log^2d
		\end{equation}
		implies that the following estimates hold simultaneously:
		\begin{align}
			\left\|\frac1N\sum_{j=1}^N(\vx^*A_j\vx)A_j
			-\E\big[(\vx^*A_j\vx)A_j\big]\right\|
			&\le\frac{\delta\zeta}{4}\|\vx\|_2^2,
			\label{eq:eq1}\\
			\left\|\frac1N\sum_{j=1}^N(A_j\vx)(A_j\vx)^*
			-\E\big[(A_j\vx)(A_j\vx)^*\big]\right\|
			&\le\frac{\delta\zeta}{4}\|\vx\|_2^2,
			\label{eq:eq2}\\
			\left\|\frac1N\sum_{j=1}^NA_j\vx(A_j\vx)^\top
			-\E\big[A_j\vx(A_j\vx)^\top\big]\right\|
			&\le\frac{\delta\zeta}{4}\|\vx\|_2^2.
			\label{eq:eq3}
		\end{align}
		The probability is at least
		\begin{equation}\label{prob}
			1-6\exp(-c_{7,\delta} d)-Nd^{-7r}.
		\end{equation}
	\end{theorem}
	
	\begin{proof}
		The claim is immediate if $\vx=0$, so assume that $\vx\ne0$. Choose
		$K_\delta>0$ sufficiently large such that
		\begin{equation}\label{eq:choice-K-delta}
			2^{-K_\delta/16}\le\frac{\delta}{128},
		\end{equation}
		and define
		\[
		\Omega_j=
		\left\{\vx^*A_j\vx
		\le K_\delta\frac rd\log d\,\|\vx\|_2^2\right\},
		\qquad
		\Omega=\bigcap_{j=1}^N\Omega_j.
		\]
		By Lemma~\ref{lem:beta-d}, the random variable
		$Y_j=\vx^*A_j\vx/\|\vx\|_2^2$ has distribution
		$\operatorname{Beta}(r,d-r)$. Write $a=r/d$. Lemma~\ref{lem:beta-concentration} gives, for
		$u\ge 2$,
		\[
		\Prob(Y_j\ge ua)
		\le\Prob(|Y_j-a|\ge (u-1)a)\le\exp(-d(u-1)a/8) \le\exp(-dua/16)\le\exp(-ur/16),%u^r\left(\frac{1-ua}{1-a}\right)^{d-r}.
		\]
		here we use $2(u-1)\ge u$ in the second last inequality. Choose $u=K_\delta\log d$ and use \eqref{eq:choice-K-delta}, we have
		\begin{equation}\label{eq:truncation-tail}
			\Prob(\Omega_j^c)\le d^{-K_\delta r/16}.
		\end{equation}
		Because $r\ge1$ and $d\ge2$,
		\begin{align*}
			d^{-K_\delta/16}=d^{-7}d^{-(K_\delta/16-7)}
			\le d^{-7}2^{-(K_\delta/16-7)}
			\le\delta d^{-7},
		\end{align*}
		The union bound and \eqref{eq:truncation-tail} also give
		\begin{equation}\label{eq:truncation-event}
			\Prob(\Omega^c)
			\le Nd^{-K_\delta r/16}
			\le \delta Nd^{-7r}\le Nd^{-7r}.
		\end{equation}
		
		Let $\mathcal N$ be a $1/4$-net of the unit sphere in $\C^d$ with
		$|\mathcal N|\le9^{2d}$. We first prove \eqref{eq:eq1}. For
		$\vz\in\mathcal N$, set
		\[
		R_j(\vz)=(\vx^*A_j\vx)(\vz^*A_j\vz),
		\qquad
		\widetilde R_j(\vz)=R_j(\vz)\mathbf1_{\Omega_j}.
		\]
		Since $0\preceq A_j\preceq I_d$, we have
		$0\le R_j(\vz)\le\|\vx\|_2^2$. Hence, 
		\begin{equation}\label{eq:truncation-bias}
			\begin{aligned}
				0\le&\E R_j(\vz)-\E\widetilde R_j(\vz)
				=\E\bigl[R_j(\vz)\mathbf1_{\Omega_j^c}\bigr]
				\le\|\vx\|^2\Prob(\Omega_j^c)\le\|\vx\|^2d^{-K_\delta r/16}\\
				\le&\|\vx\|^2d^{-K_\delta /16}\le\delta d^{-7}\|\vx\|^2\le\frac{\delta\|\vx\|^2}{16d(d+1)}
				\le\frac{\delta\zeta}{16}\|\vx\|^2.
			\end{aligned}
		\end{equation}
		For every integer $p\ge1$, Lemma~\ref{lem:beta-d} gives $\E(\vz^*A_j\vz)^p
		\le p!\left(\frac rd\right)^p$, which yields
		\begin{equation*}\label{eq:truncated-moment}
			\E|\widetilde R_j(\vz)|^p
			\le p!\left(
			K_\delta\frac{r^2}{d^2}\log d\,\|\vx\|^2
			\right)^p.
		\end{equation*}
		Since $r^2/d^2\le\lambda$, it follows that
		\begin{equation}\label{eq:psi1-bound}
			\left\|\widetilde R_j(\vz)
			-\E\widetilde R_j(\vz)\right\|_{\psi_1}
			\le C K_\delta\lambda\log d\,\|\vx\|^2
		\end{equation}
		for an absolute constant $C>0$. Bernstein's inequality now gives
		\begin{align*}
			\Prob\left(
			\left|\frac1N\sum_{j=1}^N
			\bigl(\widetilde R_j(\vz)-\E\widetilde R_j(\vz)\bigr)\right|
			>\frac{\delta\zeta}{16}\|\vx\|^2
			\right)
			\le
			2\exp(-\hat c_{7,\delta} d),
			\label{eq:scalar-bernstein}
		\end{align*}
		under $N\ge C_{7,\delta} d
		\left(\frac{\lambda}{\zeta}\right)^2\log^2d$. 
		Taking a union
		bound over $\mathcal N$, we obtain, except on an event of probability
		at most $2\exp(-c_{7,\delta}d)$,
		\begin{equation}\label{eq:uniform-truncated-X}
			\max_{\vz\in\mathcal N}
			\left|\frac1N\sum_{j=1}^N
			\bigl(\widetilde R_j(\vz)-\E\widetilde R_j(\vz)\bigr)\right|
			\le\frac{\delta\zeta}{16}\|\vx\|^2
		\end{equation}
		for some $c_{7,\delta}>0$. 
		By \eqref{eq:truncation-event}, we have
		$R_j(\vz)=\widetilde R_j(\vz)$ for every $j$ with probability at least $1-Nd^{-7r}$.
		Hence, for every
		$\vz\in\mathcal N$,  \eqref{eq:truncation-bias}
		and \eqref{eq:uniform-truncated-X} imply	    	\begin{align*}
			&\max_{\vz\in\mathcal N}\Big|\frac1N\sum_{j=1}^NR_j(\vz)-\E R_j(\vz)\Big|=\max_{\vz\in\mathcal N}\Big|\frac1N\sum_{j=1}^N\widetilde R_j(\vz)-\E R_j(\vz)\Big|\\
			\le&
			\max_{\vz\in\mathcal N}\Big|\frac1N\sum_{j=1}^N
			\bigl(\widetilde R_j(\vz)-\E\widetilde R_j(\vz)\bigr)\Big|+
			\max_{\vz\in\mathcal N}\left|\E\widetilde R_j(\vz)-\E R_j(\vz)\right|
			\le\frac{\delta\zeta}{8}\|\vx\|^2,
		\end{align*}
		with probability at least $1-Nd^{-7r}-2\exp(-c_{7,\delta}d)$.  
		The matrix in \eqref{eq:eq1} is Hermitian, so the standard net
		estimate gives by 
		Lemma \ref{appendix_lem:lem5.4}
		\[
		\Big\|\frac1N\sum_{j=1}^N(\vx^*A_j\vx)A_j
		-\E\big[(\vx^*A_j\vx)A_j\big]\Big\|
		\le2\max_{\vz\in\mathcal N}
		\Big|\frac1N\sum_{j=1}^NR_j(\vz)-\E [R_j(\vz)]\Big|\le\frac{\delta\zeta}{4}\|\vx\|^2,
		\]
		which proves \eqref{eq:eq1}.
		
		For \eqref{eq:eq2}, set
		\[
		G_j(\vz)=|\vz^*A_j\vx|^2,
		\qquad
		\widetilde G_j(\vz)=G_j(\vz)\mathbf1_{\Omega_j}. 
		\]
		Since $A_j$ is an orthogonal projection,
		\[
		G_j(\vz)
		\le(\vz^*A_j\vz)(\vx^*A_j\vx)=R_j(\vz).
		\]
		Thus $\widetilde G_j(\vz)$ satisfies the same moment and centered
		$\psi_1$ bounds as $\widetilde R_j(\vz)$. Moreover,
		$0\le G_j(\vz)\le\|\vx\|^2$ and 
		\[
		0\le\E G_j(\vz)-\E\widetilde G_j(\vz)
		\le\frac{\delta\zeta}{16}\|\vx\|^2.
		\]
		The preceding Bernstein and Hermitian net argument proves
		\eqref{eq:eq2}, apart from the probability of the failure of $\Omega$ at most $Nd^{-7r}$.
		
		Finally, for any $\vu,\vv$, set
		\[
		Q_j(\vu,\vv)=(\vu^*A_j\vx)(\vv^*A_j\vx),
		\qquad
		\widetilde Q_j(\vu,\vv)
		=Q_j(\vu,\vv)\mathbf1_{\Omega_j}.
		\]
		Cauchy--Schwarz inequality gives
		\[
		|Q_j(\vu,\vv)|
		\le\frac12
		(\vu^*A_j\vu\vx^*A_j\vx+\vv^*A_j\vv\vx^*A_j\vx).
		\]
		Since $((a+b)/2)^p\le(a^p+b^p)/2$ for $a,b\ge0$, we again have
		\[
		\E|\widetilde Q_j(\vu,\vv)|^p
		\le p!\left(
		K_\delta\frac{r^2}{d^2}\log d\,\|\vx\|^2
		\right)^p.
		\]
		The corresponding centered $\psi_1$ bound is therefore the same as \eqref{eq:psi1-bound}.
		The proof is very similar to that of \eqref{eq:eq1}, except that we apply Bernstein's inequality separately to the real and imaginary parts of $\widetilde Q_j(\vu,\vv)-\mathbb E\widetilde Q_j(\vu,\vv)$.
		
		The event $\Omega$ is common to all three estimates and is therefore
		counted only once and 
		the total failure probability is at most $
		6\exp(-c_{7,\delta} d)+Nd^{-7r}$, 
		which proves \eqref{prob}.
	\end{proof}
	
	%########################################################################
	
	%########################################################################
	%########################################################################
	%########################################################################
	\begin{lemma} \label{lem:concentration1}
		Let $A_1,A_2,\ldots,A_N$ be i.i.d. rank-$r$ orthogonal projections in $\C^d$ drawn from the Haar measure. Define the following matrix
		$$
		H(\vx):=\frac{1}{N}\sum^N_{j=1}\begin{bmatrix}
			(A_j\vx) (A_j\vx)^* & A_j\vx( A_j\vx)^\top \\
			\overline{A_j\vx}( A_j\vx)^* & \overline{A_j\vx}( A_j\vx)^\top 
		\end{bmatrix}.
		$$ Then 
		$$
		\E\big(H(\vx)\big)
		=\begin{bmatrix}
			\zeta\|\vx\|^2I+(\lambda-\zeta)\vx\vx^* & \lambda\vx\vx^\top  \\
			\lambda\overline{\vx}\vx^* & \zeta\|\vx\|^2I+(\lambda-\zeta)\overline{\vx}\vx^\top 
		\end{bmatrix}.
		$$
		For any $\delta>0$, when $N$ satisfies \eqref{N} we have with the smallest probability shown in \eqref{prob}
		$$
		\|H(\vx)-\E( H(\vx))\|\le\delta\zeta \|\vx\|^2.
		$$
	\end{lemma}
	\begin{proof} 
		The expectation is directly from Lemma \ref{lem:expectations} and the concentration inequality is from Theorem 
		\ref{theo:concentration1-1}.
	\end{proof}
	
	%########################################################################
	\begin{lemma}\label{lem:concentration1-1}
		Under the assumptions of Lemma~\ref{lem:concentration1}, on the same event, every $\vh\in\C^d$ satisfies
		\[
		\left|\frac1N\sum_{j=1}^N\operatorname{Re}^2(\vh^*A_j\vx)-\E\operatorname{Re}^2(\vh^*A\vx)\right|\le\frac{\delta\zeta}{2}\|\vx\|^2\|\vh\|^2.
		\]
		Moreover,
		\[
		\E\operatorname{Re}^2(\vh^*A\vx)=\frac\zeta2\|\vx\|^2\|\vh\|^2+\frac\mu2|\vx^*\vh|^2+\frac\lambda2\operatorname{Re}((\vx^*\vh)^2).
		\]
		If $\operatorname{Im}(\vx^*\vh)=0$, this reduces to
		\[
		\E\operatorname{Re}^2(\vh^*A\vx)=\frac\zeta2\|\vx\|^2\|\vh\|^2+\left(\lambda-\frac\zeta2\right)\operatorname{Re}^2(\vx^*\vh).
		\]
	\end{lemma}
	\begin{proof}
		Set $\mathbf w=[\vh^\top\ \overline\vh^\top]^\top$. Then $N^{-1}\sum_j\operatorname{Re}^2(\vh^*A_j\vx)=\mathbf w^*H(\vx)\mathbf w/4$, and hence
		\[
		\left|\frac1N\sum_j\operatorname{Re}^2(\vh^*A_j\vx)-\E\operatorname{Re}^2(\vh^*A\vx)\right|\le\frac14\|H(\vx)-\E H(\vx)\|\|\mathbf w\|^2,
		\]
		which gives the stated deviation bound. If $Z=\vh^*A\vx$, then $\operatorname{Re}^2Z=(|Z|^2+\operatorname{Re}Z^2)/2$. Lemma~\ref{lem:expectations} gives
		\[
		\E|Z|^2=\vh^*\E[A\vx(A\vx)^*]\vh=\zeta\|\vx\|^2\|\vh\|^2+\mu|\vx^*\vh|^2,\qquad \E Z^2=\vh^*\E[A\vx(A\vx)^\top]\overline{\vh}=\lambda(\vx^*\vh)^2,
		\]
		which proves the expectation formulas.
	\end{proof}
	%%----------------------------------------------------------------------------	
	\begin{lemma}\label{lem:initial_bound}
		Let $A$ be a rank-$r$ orthogonal projection in $\C^d$ distributed from the Haar measure, where $1\le r<d$. Let $\vx,\vh\in\C^d$ satisfy $
		\operatorname{Im}(\vx^*\vh)=0,\ \|\vh\|=1,
		$
		and define
		$
		\varepsilon_0=\sqrt{\frac{10\zeta}{27\lambda}},\ 
		\varepsilon_1=\sqrt{\frac{10}{27}}.
		$
		Then, for every $0\le s\le\varepsilon_0\|\vx\|$,
		\[
		2\E\!\left[\operatorname{Re}^2(\vh^*A\vx)\right]
		+3s\E\!\left[\operatorname{Re}(\vh^*A\vx)\,\vh^*A\vh\right]
		+\frac{9}{10}s^2\E\!\left[(\vh^*A\vh)^2\right]
		\ge\frac{\zeta}{2}\|\vx\|^2.
		\]
		Moreover, for every $0\le s\le\varepsilon_1\|\vx\|$,
		\begin{align*}
			2\E\!\left[\operatorname{Re}^2(\vh^*A\vx)\right]
			+3s\E\!\left[\operatorname{Re}(\vh^*A\vx)\,\vh^*A\vh\right]
			+&\left(1-\frac{\zeta}{10\lambda}\right)s^2\E\!\left[(\vh^*A\vh)^2\right]\\
            -&\mu\left(s^2+3s\operatorname{Re}(\vx^*\vh)
			+2\operatorname{Re}^2(\vx^*\vh)\right)
			\ge\frac{\zeta}{2}\|\vx\|^2.
		\end{align*}
	\end{lemma}
	
	\begin{proof}
		Set $a=\operatorname{Re}(\vx^*\vh)$. Since $\operatorname{Im}(\vx^*\vh)=0$, Lemma~\ref{lem:expectations} and \ref{lem:concentration1-1} give
		\[
		\E\!\left[\operatorname{Re}^2(\vh^*A\vx)\right]
		=\frac{\zeta}{2}\|\vx\|^2+\left(\lambda-\frac{\zeta}{2}\right)a^2, \qquad\E\!\left[\operatorname{Re}(\vh^*A\vx)\,\vh^*A\vh\right]=\lambda a,
		\qquad
		\E\!\left[(\vh^*A\vh)^2\right]=\lambda.
		\]
		Therefore,
		\begin{align*}
			&2\E\!\left[\operatorname{Re}^2(\vh^*A\vx)\right]
			+3s\E\!\left[\operatorname{Re}(\vh^*A\vx)\,\vh^*A\vh\right]
			+\frac{9}{10}s^2\E\!\left[(\vh^*A\vh)^2\right]\\
			=&\zeta\|\vx\|^2+(2\lambda-\zeta)a^2+3s\lambda a+\frac{9}{10}s^2\lambda\ge\zeta\|\vx\|^2+\lambda\left(a^2-3s|a|+\frac{9}{10}s^2\right)\\
			=&\zeta\|\vx\|^2+\lambda\left(|a|-\frac{3s}{2}\right)^2-\frac{27}{20}\lambda s^2\ge\zeta\|\vx\|^2-\frac{27}{20}\lambda s^2.
		\end{align*}
		Here we used $2\lambda-\zeta\ge\lambda$. If $s\le\varepsilon_0\|\vx\|$, then
		\[
		\frac{27}{20}\lambda s^2\le\frac{\zeta}{2}\|\vx\|^2,
		\]
		which proves the first estimate.
		
		For the second estimate, using $\zeta=\lambda-\mu$, we obtain
		\begin{align*}
			&2\E\!\left[\operatorname{Re}^2(\vh^*A\vx)\right]
			+3s\E\!\left[\operatorname{Re}(\vh^*A\vx)\,\vh^*A\vh\right]
			+\left(1-\frac{\zeta}{10\lambda}\right)s^2\E\!\left[(\vh^*A\vh)^2\right]-\mu\left(s^2+3sa+2a^2\right)\\
			=&\zeta\|\vx\|^2+\zeta a^2+3s\zeta a+\frac{9}{10}s^2\zeta\ge\zeta\|\vx\|^2+\zeta\left(a^2-3s|a|+\frac{9}{10}s^2\right)\\
			=&\zeta\|\vx\|^2+\zeta\left(|a|-\frac{3s}{2}\right)^2-\frac{27}{20}\zeta s^2\ge\zeta\|\vx\|^2-\frac{27}{20}\zeta s^2.
		\end{align*}
		If $s\le\varepsilon_1\|\vx\|$, then
		\[
		\frac{27}{20}\zeta s^2\le\frac{\zeta}{2}\|\vx\|^2,
		\]
		and the second estimate follows.
	\end{proof}
	%%-------------------------------------------------------------------
	\begin{lemma}\label{lem:estimate_diff}
		Let $A_1,\ldots,A_N$ be rank-$r$ orthogonal projections in $\C^d$, and define 
		\[
		p(\vu,s)=\frac1N\sum_{j=1}^N
		\left(\sqrt{\frac{9}{4\iota}}\operatorname{Re}(\vu^*A_j\vx)
		+\sqrt{\iota}s\,\vu^*A_j\vu\right)^2,\ \forall\iota\in(0,1].
		\]
		Let $s\ge0$ and let $\vu,\vv\in\C^d$ be unit vectors. If $
		\left\|\frac1N\sum_{j=1}^NA_j\right\|
		\le (1+\delta)\frac rd$ 
		for some $\delta>0$, then
		\[
		|p(\vu,s)-p(\vv,s)|
		\le
		2\left(\sqrt{\frac{9}{4\iota}}\|\vx\|+\sqrt{\iota}s\right)
		\left(\sqrt{\frac{9}{4\iota}}\|\vx\|+2\sqrt{\iota}s\right)
		(1+\delta)\frac rd\,\|\vu-\vv\|.
		\]
	\end{lemma}
	
	\begin{proof}
		Set
		\[
		a_j=\sqrt{\frac{9}{4\iota}}\operatorname{Re}(\vu^*A_j\vx)
		+\sqrt{\iota}s\,\vu^*A_j\vu,\ 
		b_j=\sqrt{\frac{9}{4\iota}}\operatorname{Re}(\vv^*A_j\vx)
		+\sqrt{\iota}s\,\vv^*A_j\vv.
		\]
		Since $A_j$ is an orthogonal projection and $\vu,\vv$ are unit vectors,
		\[
		|a_j|,|b_j|
		\le \sqrt{\frac{9}{4\iota}}\|\vx\|+\sqrt{\iota}s.
		\]
		Therefore,
		\[
		|p(\vu,s)-p(\vv,s)|
		\le\frac{2}{N}\left(\sqrt{\frac{9}{4\iota}}\|\vx\|+\sqrt{\iota}s\right)
		\sum_{j=1}^N|a_j-b_j|.
		\]
		Moreover,
		\[
		|a_j-b_j|
		\le \sqrt{\frac{9}{4\iota}}\left|(\vu-\vv)^*A_j\vx\right|
		+\sqrt{\iota}s\left|\vu^*A_j\vu-\vv^*A_j\vv\right|.
		\]
		
		For the first term, the Cauchy--Schwarz inequality gives
		\begin{align*}
			\frac1N\sum_{j=1}^N\left|(\vu-\vv)^*A_j\vx\right|
			&\le
			\left(\frac1N\sum_{j=1}^N(\vu-\vv)^*A_j(\vu-\vv)\right)^{1/2}
			\left(\frac1N\sum_{j=1}^N\vx^*A_j\vx\right)^{1/2}\le
			\left\|\frac1N\sum_{j=1}^NA_j\right\|
			\|\vu-\vv\|\|\vx\|.
		\end{align*}
		Also,
		\[
		\vu^*A_j\vu-\vv^*A_j\vv
		=(\vu-\vv)^*A_j\vu+\vv^*A_j(\vu-\vv).
		\]
		Hence, applying the same argument twice,
		\[
		\frac1N\sum_{j=1}^N
		\left|\vu^*A_j\vu-\vv^*A_j\vv\right|
		\le
		2\left\|\frac1N\sum_{j=1}^NA_j\right\|\|\vu-\vv\|.
		\]
		Combining the above estimates yields
		\begin{align*}
			|p(\vu,s)-p(\vv,s)|
			&\le
			2\left(\sqrt{\frac{9}{4\iota}}\|\vx\|+\sqrt{\iota}s\right)
			\left(\sqrt{\frac{9}{4\iota}}\|\vx\|+2\sqrt{\iota}s\right)\cdot
			\left\|\frac1N\sum_{j=1}^NA_j\right\|\|\vu-\vv\|.
		\end{align*}
		The stated result follows from the assumed bound on $\|N^{-1}\sum_{j=1}^NA_j\|$.
	\end{proof}
	\begin{lemma}\label{lem:expectation_u}
		Let $A_1,\ldots,A_N$ be i.i.d rank-$r$ orthogonal projections in $\C^d$ drawn from the Haar measure, where $1\le r<d$.
		Fix $\vx\ne0$. For $\iota\in[9/10,1]$, $0\le s\le\varepsilon_1\|\vx\|$, and a fixed $\vh\in\C^d$ with $\|\vh\|=1$, define
		\begin{align*}
			U_j : &= \frac{9}{4\iota}\big(\re(\vh^*A_j\vx)\big)^2
			+3s\re(\vh^*A_j\vx)(\vh^*A_j\vh)+\iota s^2|\vh^*A_j\vh|^2,j=1,\dots,N .
		\end{align*} 
		We have when $N\geq Cd\left(\frac{\lambda}{\zeta}\right)^2\log^2d$
		\[
		\Prob\bigg(N\E U_j - \sum_{j=1}^{N}U_j\geq\frac{\zeta}{8}N\|\vx\|^2 \bigg)\leq e^{-3\gamma d\log^2d},
		\]
		where $C,\gamma>0$ are constants independent of $s$, $\iota$, and $\vh$ (and may be chosen uniformly on the displayed ranges).
	\end{lemma}
	\begin{proof}

		Since $ A_j, j=1,\ldots,N $ are i.i.d., we have $\E U_1=\ldots=\E U_N$. Then we have for a fixed unit vector $ \vh $
		\begin{align*}
			u := \E U_j
			\leq& \frac{9}{4\iota}\sqrt{\E\big(|\vh^*A_j\vh|^2\big)\E\big(|\vx^*A_j\vx|^2\big)}
			+3s\sqrt{\E|\vh^*A_j\vx|^2\,\E|\vh^*A_j\vh|^2}
			+\iota s^2\E(|\vh^*A_j\vh|^2)\\
			\le&\left(\frac{9}{4\iota}\|\vh\|^2\|\vx\|^2+3s\|\vh\|^3\|\vx\|+\iota s^2\|\vh\|^4\right)\cdot \lambda:=b(t,\iota)\lambda\|\vx\|^2,
		\end{align*}
		where $t=s/\|\vx\|$ and $b(t,\iota)=\frac{9}{4\iota}\|\vh\|^2+3t\|\vh\|^3+\iota t^2\|\vh\|^4$. Now define $ X_j=u - U_j,j=1,\ldots,N $, we bound $ \E X_j ^2,j=1,\ldots,N$ using Holder's inequality. Since $ U_j\geq 0 $, we have $ X_j\leq u\leq b(t,\iota)\lambda\|\vx\|^2 $ and 
		\begin{align*}
			\E X_j^2\leq \E U_j^2=&\frac{81}{16\iota^2}\E\big(\re^4(\vh^*A_j\vx)\big)+\iota^2s^4\E\big(|\vh^*A_j\vh|^4\big)+\frac{27}{2}s^2\E\Big(\re^2(\vh^*A_j\vx)|\vh^*A_j\vh|^2\Big)\\
			&+\frac{27}{2\iota}s\E\Big(\re^3(\vh^*A_j\vx)|\vh^*A_j\vh|\Big)+6\iota s^3\E\Big(\re(\vh^*A_j\vx)|\vh^*A_j\vh|^3\Big)\\
			\le&\frac{81}{16\iota^2}\sqrt{\E\big(|\vh^*A_j\vh|^4\big)\E\big(|\vx^*A_j\vx|^4\big)}+\iota^2s^4\E\big(|\vh^*A_j^*\vh|^4\big)+\frac{27s^2}{2}\sqrt{\E\big(|\vh^*A_j^*\vh|^{6}\big)\E\big(|\vx^*A_j^*\vx|^2\big)}\\
			&+\frac{27}{2\iota}s\sqrt{\E\big(|\vh^*A_j\vh|^{5}\big)\E\big(|\vx^*A_j\vx|^3\big)}+6\iota s^3\sqrt{\E\big(|\vh^*A_j\vh|^{7}\big)\E\big(|\vx^*A_j\vx|\big)}  \\
			\leq& \bigg(\frac{81}{16\iota ^2}\cdot \Gamma(5)\|\vx\|^4+\iota ^2s^4\cdot \Gamma(5) +\frac{27s^2}{2}\cdot \sqrt{\Gamma(7)\cdot \Gamma(3)}\|\vx\|^2\\
			&+\frac{27}{2\iota }s\cdot \sqrt{ \Gamma(6)\cdot \Gamma(4)}\|\vx\|^3+6\iota  s^3\cdot  \sqrt{ \Gamma(8)\cdot \Gamma(2)}\|\vx\|\bigg) \cdot \left(\frac{r}{d}\right)^4\le C(t,\iota)\lambda^2\|\vx\|^4,
		\end{align*} 
		where the last inequality is from Lemma \ref{lem:beta-d} and $C(t,\iota)=\frac{81}{16\iota ^2}\cdot \Gamma(5)+\iota ^2t^4\cdot \Gamma(5) +\frac{27t^2}{2}\cdot \sqrt{\Gamma(7)\cdot \Gamma(3)}+\frac{27}{2\iota }t\cdot \sqrt{ \Gamma(6)\cdot \Gamma(4)}+6\iota  t^3\cdot  \sqrt{ \Gamma(8)\cdot \Gamma(2)}$ with $t=s/\|\vx\|$ and $(t,\iota)\in[0,\varepsilon_1]\times[9/10,1]$. 
		Set
		\[
		D_*:=\sup_{0\le t\le\varepsilon_1,\,9/10\le\iota\le1}\max\{b(t,\iota)^2,C(t,\iota)\}<\infty.
		\]
		Then
		\[
		X_j\le u\le b(t,\iota)\lambda\|\vx\|^2\le \sqrt{D_*}\lambda\|\vx\|^2.
		\qquad
		\E X_j^2\le D_*\|\vx\|^4\lambda^2.
		\]
		Choose $y=N\zeta\|\vx\|^2/8$ and put $\sigma^2\le ND_*\lambda^2\|\vx\|^4$.  Since $\zeta\le\lambda$, Lemma~\ref{lem:bentkus_one_sided} gives the explicit uniform estimate
		\[
		\Prob\!\left(Nu-\sum_{j=1}^NU_j\ge\frac{\zeta}{8}N\|\vx\|^2\right)
		\le\exp\!\left[-\frac{N}{128D_*
		}\left(\frac{\zeta}{\lambda}\right)^2\right].
		\]
		Thus, when $N\geq Cd\left(\frac{\lambda}{\zeta}\right)^2\log^2d$,
		\begin{align*}
			\Prob\bigg(Nu - \sum_{j=1}^{N}U_j\geq\frac{\zeta}{8}N\|\vx\|^2 \bigg)\leq e^{-3\gamma d\log^2d},
		\end{align*}
		where one may take any fixed $0<\gamma\le(384D_*
		)^{-1}$. 
	\end{proof}
	%#################################
	\begin{lemma}\label{lem:square_estiate}
		Let $A_1,\ldots,A_N$ be i.i.d rank-$r$ orthogonal projections in $\C^d$ drawn from the Haar measure, where $1\le r< d$. Let $\vx\ne0$, and assume $N\ge C d(\lambda/\zeta)^2\log^2d$  with an appropriate positive constant $C>0$. Then there exists a constant $\gamma>0$, with probability at least $1-\exp(-\gamma d\log^2d)$, the following inequalities hold 
		
		\begin{align}
			&\frac{1}{N}\sum^N_{j=1}\Bigg( \sqrt{\frac{5}{2}}\re(\vh^*A_j\vx)+\sqrt{\frac{9}{10}}|\vh^*A_j\vh|\Bigg)^2\ge\frac{\zeta}{4}\|\vx\|^2\|\vh\|^2+\frac{ \E\big(\re^2(\vh^*A_j\vx)\big)}{2},
			\ \forall \vh\in\calS_1,\label{curvature_condition1}\\
			&\frac{1}{N}\sum_{j=1}^{N}\Bigg( \frac{3}{2\sqrt{1-\frac{\zeta}{10\lambda}}}\re(\vh^*A_j\vx)+\sqrt{1-\frac{\zeta}{10\lambda}}|\vh^*A_j\vh|\Bigg)^2-\mu(\|\vh\|^4+3\re(\vx^*\vh)\|\vh\|^2+2(\re(\vx^*\vh))^2)\nonumber\\
			\ge&\left(\frac{9}{4\left(1-\frac{\zeta}{10\lambda}\right)}-2\right)\E\big(\re^2(\vh^*A_j\vx)\big)+\frac{\zeta}{4}\|\vx\|^2\|\vh\|^2,\ \forall \vh\in\calS_2,\label{curvature_condition2}
		\end{align}
		where $\calS_1=\{\vh:0<\|\vh\|\le\varepsilon_0\|\vx\|,\ \im(\vh^*\vx)=0\}$ and $\calS_2=\{\vh:\varepsilon_0\|\vx\|/2\le\|\vh\|\le\varepsilon_1\|\vx\|,\ \im(\vh^*\vx)=0\}$.
	\end{lemma}
	
	\begin{proof}[Proof of Inequality \eqref{curvature_condition1}] 
		Set $s=\|\vh\|$. Equivalently, we only need to prove that for all  $ \vh $ satisfying $\im(\vh^*\vx)=0 $, $ \|\vh\|=1 $ and for all $ s $ with $ 0< s\le\varepsilon_0\|\vx\| $, we have 
		\begin{equation}\label{curvature_condition1-1}
			\frac{1}{N}\sum^N_{j=1}\left(\frac{5}{2}\big(\re(\vh^*A_j\vx)\big)^2
			+3s\re(\vh^*A_j\vx)(\vh^*A_j\vh)+\frac{9s^2}{10}|\vh^*A_j\vh|^2\right)\ge\frac{\zeta}{4}\|\vx\|^2+\frac{ \E\big(\re^2(\vh^*A_j\vx)\big)}{2}.
		\end{equation}  
		We first fix $0<s<\varepsilon_0\|\vx\|$ and $\vh\in\C^d$ satisfying $\im(\vh^*\vx)=0 $, $ \|\vh\|=1 $. Set 
		\begin{align*}
			U_j : =\Bigg( \sqrt{\frac{5}{2}}\re(\vh^*A_j\vx)+\sqrt{\frac{9}{10}}s|\vh^*A_j\vh|\Bigg)^2= \frac{5}{2}\big(\re(\vh^*A_j\vx)\big)^2
			+3s\re(\vh^*A_j\vx)(\vh^*A_j\vh)+\frac{9s^2}{10}|\vh^*A_j\vh|^2.
		\end{align*}                                                                                  
		Since $ A_j, j=1,\ldots,N $ are i.i.d., we have $\E U_1=\ldots=\E U_N:=u$. By Lemma \ref{lem:expectation_u} with $\iota=\frac{9}{10}$, we have 
		\[
		\Prob\bigg(Nu - \sum_{j=1}^{N}U_j\geq\frac{\zeta}{8}N\|\vx\|^2 \bigg)\leq e^{-3\gamma d\log^2d},
		\]
		Here $\gamma>0$ is the uniform constant from Lemma~\ref{lem:expectation_u}. 
		Therefore, with probability at least $ 1-\exp(-3\gamma d\log^2d) $, we have
		\begin{equation}
			\label{curvature_cond_fixh}
			\begin{split}
				&\frac{1}{N}\sum_{j=1}^{N}U_j\geq u- \frac{\zeta}{8}\|\vx\|^2\\
				=&\frac{1}{2}\E\big(\re^2(\vh^*A_j\vx)\big)+2\E(\re^2(\vh^*A_j\vx)) + 3s\E(\re(\vh^*A_j\vx)|\vh^*A_j\vh|)+\frac{9}{10}s^2\E(|\vh^*A_j\vh|^2)-\frac{\zeta}{8}\|\vx\|^2\\
				\geq& \frac{1}{2}\E\big(\re^2(\vh^*A_j\vx)\big)+\frac{\zeta}{2}\|\vx\|^2-\frac{\zeta}{8}\|\vx\|^2
				\geq	 \frac{1}{2}\E\big(\re^2(\vh^*A_j\vx)\big)+\frac{3\zeta}{8}\|\vx\|^2.
			\end{split}
		\end{equation}
		Here the second inequality comes from Lemma \ref{lem:initial_bound}, which requires $0<s\le\varepsilon_0\|\vx\|$. 
		
		\noindent
		Next, we consider a fixed $0< s\leq\varepsilon_0\|\vx\| $ and all unit vectors $\vh$ with $\im(\vh^*\vx)=0$. 
		Define $$p(\vu,s)=\frac{1}{N}\sum_{j=1}^{N}\Bigg( \sqrt{\frac{5}{2}}\re(\vu^*A_j\vx)+\sqrt{\frac{9}{10}}s|\vu^*A_j\vu|\Bigg)^2.$$ For any unit vectors $\vu,\vv$,  we use Lemma 
		\ref{lem:estimate_diff} with $\iota=\frac{9}{10}$ and get 
		\begin{align*}
			|p(\vu,s)-p(\vv,s)|
			\le& 2\left(\sqrt{\frac{5}{2}}\|\vx\|+\sqrt{\frac{9}{10}}s\right)\left( \sqrt{\frac{5}{2}}\|\vx\|+2\sqrt{\frac{9}{10}}s\right)(1+\delta)\frac{r}{d}\|\vu-\vv\|:=\chi\frac{r}{d}\|\vu-\vv\|.
		\end{align*}
		Thus for $\forall \vu,\vv\in\bbS^{d-1}$ with $ \|\vu-\vv\|\leq\eta:=\frac{\zeta d\|\vx\|^2}{32r\chi} $, we have
		\begin{equation}\label{covering_cond1}
			p(\vv,s)\ge p(\vu,s)-\frac{\zeta}{32}\|\vx\|^2.
		\end{equation}
		Set $\mathbb T_\vx:=\{\vu\in\C^d:\|\vu\|=1,\ \im(\vu^*\vx)=0\}\subseteq\bbS^{d-1}$ and choose an $\eta$-net $\calN_\eta\subset\mathbb T_\vx$ with $|\calN_\eta|\le(1+2/\eta)^{2d}$. %}
	For a fixed $ s\leq\varepsilon_0\|\vx\| $, we apply \eqref{curvature_cond_fixh} to all $ \vu\in\calN_\eta $ and get 
	\begin{equation}\label{covering_cond2}
		\begin{split}
			\Prob\,\Bigg(p(\vu,s)\geq  \frac{1}{2}\E\big(\re^2(\vu^*A_j\vx)\big)+\frac{3\zeta}{8}\|\vx\|^2\Bigg)\ge 1-|\calN_\eta|\exp(-3\gamma d\log^2d)
			\geq 1-\exp(-2\gamma d\log^2d),
		\end{split}
	\end{equation}
	where the last inequality holds when $N$ satisfies \eqref{N}. In addition, we have by Lemma \ref{lem:concentration1-1} and $\|\vu\|=\|\vh\|=1$
	\begin{align}\label{eq:diff-expectation}
		\frac{\E\big(\re^2(\vu^*A\vx)\big)-\E\big(\re^2(\vh^*A\vx)\big)}{2}
		=&\frac{\zeta}{4}\|\vx\|^2(\|\vu\|^2-\|\vh\|^2)+\frac{1}{2}\left(\lambda-\frac{\zeta}{2}\right)\re(\vx^*(\vu+\vh))\re(\vx^*(\vu-\vh))\nonumber\\
		\ge&-\left(\lambda-\frac{\zeta}{2}\right)\|\vx\|^2\|\vu-\vh\|\ge-\left(\lambda-\frac{\zeta}{2}\right)\|\vx\|^2\eta\ge-\frac{\zeta\|\vx\|^2}{32}.
	\end{align}
	For a fixed $s$, let $\vh\in\mathbb T_\vx$ and choose $\vu\in\calN_\eta$ with $\|\vh-\vu\|\le\eta$.  Equations~\eqref{covering_cond1}, \eqref{covering_cond2}, and \eqref{eq:diff-expectation} give
	\begin{align}\label{eq:all-h}
		p(\vh,s)
		&\ge p(\vu,s)-\frac{\zeta}{32}\|\vx\|^2\ge \frac12\E\re^2(\vu^*A\vx)+\frac{11\zeta}{32}\|\vx\|^2\ge \frac12\E\re^2(\vh^*A\vx)+\frac{10\zeta}{32}\|\vx\|^2
	\end{align}
	with probability at least $ 1-\exp(-2\gamma d\log^2d) $. 
	Next, we consider all $\vu\in\mathbb T_\vx$ and all $0\le s\le\varepsilon_0\|\vx\|$. Applying a similar covering number argument over $ 0< s\leq \varepsilon_0\|\vx\| $ and combining \eqref{eq:all-h}, we can further get for all $ 0< s\leq\varepsilon_0\|\vx\| $ and all $ \vh\in\mathbb T_\vx $  
	$$ 
	p(\vh,s)\geq  \frac{\E\big(\re^2(\vh^*A\vx)\big)}{2}+\frac{\zeta}{4} \|\vx\|^2
	$$ 
	holds with probability at least $ 1-\exp(-\gamma d\log^2d) $. 
\end{proof}
\begin{proof}[Proof of Inequality of \eqref{curvature_condition2}]
	The proof is very similar as that of the proof for Inequality \eqref{curvature_condition1}. Set $s=\|\vh\|$. Equivalently, we only need to prove that for all  $ \vh $ satisfying $\im(\vh^*\vx)=0 $, $ \|\vh\|=1 $ and for all $ s $ with $ \frac{\varepsilon_0}2\|\vx\|\le s\le\varepsilon_1\|\vx\| $, we have
	\begin{equation}\label{eq:curvature_condition1-1}
		\begin{split}
			&\frac{1}{N}\sum^N_{j=1}\left(\frac{9}{4\left(1-\frac{\zeta}{10\lambda}\right)}\big(\re(\vh^*A_j\vx)\big)^2
			+3s\re(\vh^*A_j\vx)(\vh^*A_j\vh)+\left(1-\frac{\zeta}{10\lambda}\right)s^2|\vh^*A_j\vh|^2\right)\\
			&-\mu(s^2\|\vh\|^2+3s\re(\vx^*\vh)\|\vh\|+2(\re(\vx^*\vh))^2)\ge\frac{\zeta}{4}\|\vx\|^2+\left(\frac{9}{4\left(1-\frac{\zeta}{10\lambda}\right)}-2\right) \E\big(\re^2(\vh^*A_j\vx)\big).
		\end{split}
	\end{equation}
	We first fix $\vh\in\C^d$ and  $\frac{\varepsilon_0}2\|\vx\|\le  s\le\varepsilon_1\|\vx\|$. Set 
	\begin{align*}
		U_j : &= \frac{9}{4\left(1-\frac{\zeta}{10\lambda}\right)}\big(\re(\vh^*A_j\vx)\big)^2
		+3s\re(\vh^*A_j\vx)(\vh^*A_j\vh)+\left(1-\frac{\zeta}{10\lambda}\right)s^2|\vh^*A_j\vh|^2.
	\end{align*}                                                                                
	Since $ A_j, j=1,\ldots,N $ are i.i.d., we have $\E U_1=\ldots=\E U_N:=u$.  By Lemma \ref{lem:expectation_u} with $\iota=1-\frac{\zeta}{10\lambda}$, we have 
	\begin{align*}
		\Prob\bigg(Nu - \sum_{j=1}^{N}U_j\geq\frac{\zeta}{8}N\|\vx\|^2 \bigg)\leq e^{-3\gamma d\log^2d}.
	\end{align*}
	Here $\gamma>0$ is the uniform constant provided by Lemma~\ref{lem:expectation_u}. 
	Therefore, with probability at least $ 1-\exp(-3\gamma d\log^2d) $, we have
	\begin{align}
		&\frac{1}{N}\sum_{j=1}^{N}\big(U_j-\mu(s^2\|\vh\|^2+3s\re(\vx^*\vh)+2(\re(\vx^*\vh))^2)\big)\nonumber\\
		\geq& (u-\mu(s^2\|\vh\|^2+3s\re(\vx^*\vh)+2(\re(\vx^*\vh))^2))- \frac{\zeta}{8}\|\vx\|^2\nonumber\\
		=&\left(\frac{9}{4\left(1-\frac{\zeta}{10\lambda}\right)}-2\right)\E\big(\re^2(\vh^*A_j\vx)\big)+2\E(\re^2(\vh^*A_j\vx)) + 3s\E(\re(\vh^*A_j\vx)|\vh^*A_j\vh|)\nonumber\\
		&\hspace{1.5cm}+\left(1-\frac{\zeta}{10\lambda}\right)s^2\E(|\vh^*A_j\vh|^2)-\mu(s^2\|\vh\|^2+3s\re(\vx^*\vh)+2(\re(\vx^*\vh))^2)-\frac{\zeta}{8}\|\vx\|^2\nonumber\\
		\geq& \left(\frac{9}{4\left(1-\frac{\zeta}{10\lambda}\right)}-2\right)\E\big(\re^2(\vh^*A_j\vx)\big)+\frac{\zeta}{2}\|\vx\|^2-\frac{\zeta}{8}\|\vx\|^2
		\geq	 \left(\frac{9}{4\left(1-\frac{\zeta}{10\lambda}\right)}-2\right)\E\big(\re^2(\vh^*A_j\vx)\big)+\frac{3\zeta}{8}\|\vx\|^2.\label{eq:fix-all}
	\end{align}
	Here the second inequality comes from Lemma \ref{lem:initial_bound}, which requires $0<s\le\varepsilon_1\|\vx\|$. 
	
	\noindent
	Next, we consider a fixed $s\in[\frac{\varepsilon_0}2\|\vx\|, \varepsilon_1\|\vx\|] $ and all unit vectors $\vh$ with $\im(\vh^*\vx)=0$. 
	Define 
	\begin{align*}
		p(\vu,s)=&\frac{1}{N}\sum_{j=1}^{N}\Bigg( \sqrt{\frac{9}{4\left(1-\frac{\zeta}{10\lambda}\right)}}\re(\vu^*A_j\vx)+\sqrt{1-\frac{\zeta}{10\lambda}}s|\vu^*A_j\vu|\Bigg)^2,\\ q(\vu,s)=&p(\vu,s)-\mu(s^2\|\vu\|^2+3s\re(\vx^*\vu)+2(\re(\vx^*\vu))^2).
	\end{align*} For any unit vectors $\vu,\vv$,  we use Lemma 
	\ref{lem:estimate_diff}  with $\iota=1-\frac{\zeta}{10\lambda}$ and get $
	|p(\vu,s)-p(\vv,s)|
	\le 2\kappa(1+\delta)\frac{r}{d}\|\vu-\vv\|,$
	where $\kappa=\left(\sqrt{\frac{9}{4\left(1-\frac{\zeta}{10\lambda}\right)}}\|\vx\|+\sqrt{1-\frac{\zeta}{10\lambda}}s\right)\left( \sqrt{\frac{9}{4\left(1-\frac{\zeta}{10\lambda}\right)}}\|\vx\|+2\sqrt{1-\frac{\zeta}{10\lambda}}s\right)$. Then
	\begin{align*}
		|q(\vu,s)-q(\vv,s)|
		\le& |p(\vu,s)-p(\vv,s)|+3s\mu|\re(\vx^*\vu)-\re(\vx^*\vv)|+2\mu|(\re(\vx^*\vu))^2-(\re(\vx^*\vv))^2|\\
		\le&
		\left[2\kappa(1+\delta)\frac{r}{d}+3s\mu\|\vx\|+4\|\vx\|^2\mu\right]\|\vu-\vv\|\\
		\le&
		\left[2\kappa(1+\delta)+3s\|\vx\|+4\|\vx\|^2\right]\frac{r}{d}\|\vu-\vv\|
		:=\chi\frac{r}{d}\|\vu-\vv\|.
	\end{align*}
	Thus for $\forall \vu,\vv\in\bbS^{d-1}$ with $ \|\vu-\vv\|\leq\eta:=\frac{\zeta d\|\vx\|^2}{32r\chi}$, we have
	\begin{equation}\label{covering_cond1_tilde}
		q(\vv,s)\ge q(\vu,s)-\frac{\zeta}{32}\|\vx\|^2.
	\end{equation}
	Set $\calN_\eta$ to be an $\eta$-net of $\mathbb T_\vx:=\{\vu\in\C^d:\|\vu\|=1,\ \im(\vu^*\vx)=0\}$. 
	For a fixed $\frac{\varepsilon_0}{2}\|\vx\|\le s\leq\varepsilon_1\|\vx\| $, we apply \eqref{eq:fix-all} to all $ \vu\in\calN_\eta $ and get 
	\begin{equation}\label{covering_cond2_tilde}
		\begin{split}
			\Prob\,\Bigg(q(\vu,s)\geq  \left(\frac{9}{4\left(1-\frac{\zeta}{10\lambda}\right)}-2\right)\E\big(\re^2(\vu^*A_j\vx)\big)+\frac{3\zeta}{8}\|\vx\|^2\Bigg)\ge& 1-|\calN_\eta|\exp(-3\gamma d\log^2d)\\
			\geq& 1-\exp(-2\gamma d\log^2d),
		\end{split}
	\end{equation}
	where the last inequality holds when $N$ satisfies \eqref{N}. In addition, we have by Lemma \ref{lem:concentration1-1}
	\begin{align}\label{eq:diff-expectation1}
		&\left(\frac{9}{4\left(1-\frac{\zeta}{10\lambda}\right)}-2\right)\left[\E\big(\re^2(\vu^*A\vx)\big)-\E\big(\re^2(\vh^*A\vx)\big)\right]\nonumber\\
		=&\left(\frac{9}{4\left(1-\frac{\zeta}{10\lambda}\right)}-2\right)\frac{\zeta}{2}\|\vx\|^2(\|\vu\|^2-\|\vh\|^2)+\left(\frac{9}{4\left(1-\frac{\zeta}{10\lambda}\right)}-2\right)\left(\lambda-\frac{\zeta}{2}\right)\re(\vx^*(\vu+\vh))\re(\vx^*(\vu-\vh))\nonumber\\
		\ge&-\left(\lambda-\frac{\zeta}{2}\right)\|\vx\|^2\|\vu-\vh\|\ge-\left(\lambda-\frac{\zeta}{2}\right)\|\vx\|^2\eta\ge-\frac{\zeta\|\vx\|^2}{32},
	\end{align}
	since $\|\vu\|=\|\vh\|=1$ and $0<\left(\frac{9}{4\left(1-\frac{\zeta}{10\lambda}\right)}-2\right)\le\frac{1}{2}$. 
	Set
	\[
	c_*:=\frac{9}{4(1-\zeta/(10\lambda))}-2.
	\]
	For a fixed $s$, choose $\vu\in\calN_\eta$ with $\|\vh-\vu\|\le\eta$.  Equations~\eqref{covering_cond1_tilde}, \eqref{covering_cond2_tilde}, and \eqref{eq:diff-expectation1} yield
	\begin{align}\label{eq:all-h-q}
		q(\vh,s)
		&\ge q(\vu,s)-\frac{\zeta}{32}\|\vx\|^2\ge c_*\E\re^2(\vu^*A\vx)+\frac{11\zeta}{32}\|\vx\|^2\ge c_*\E\re^2(\vh^*A\vx)+\frac{10\zeta}{32}\|\vx\|^2
	\end{align}
	with probability at least $1-\exp(-2\gamma d\log^2d)$. 
	
	\noindent
	Finally, we consider all $s\in[\frac{\varepsilon_0}{2}\|\vx\|,\varepsilon_1\|\vx\|]$ and all $\vh\in\mathbb T_{\vx}$. Applying a similar covering number argument over $ \frac{\varepsilon_0}{2}\|\vx\|\le s\leq \varepsilon_1\|\vx\| $ and combining \eqref{eq:all-h-q}, we can further get for all $ \frac{1}{2}\varepsilon_0\|\vx\|\le s\leq\varepsilon_1\|\vx\| $ and all $ \vh\in\mathbb T_{\vx} $, 
	$$ 
	q(\vh,s)\geq  \left(\frac{9}{4\left(1-\frac{\zeta}{10\lambda}\right)}-2\right)\E\big(\re^2(\vh^*A\vx)\big)+\frac{\zeta}{4} \|\vx\|^2
	$$ 
	holds with probability at least $ 1-\exp(-\gamma d\log^2d) $.
\end{proof}

%#########################################

\appendix
\section{Definitions and Lemmas}
\label{sec:appendix}
\begin{definition}[Nets and covering numbers,  {\cite[Definition~5.1]{vershynin2010introduction}}]
	Let $(X,d)$ be a metric space and $\eta>0$. A subset $\calN_\eta$ of $X$ is called an $\eta$-net of $X$ if every point $\vv\in X$ can be approximated to within $\eta$ by some point $\vu\in\calN_\eta$, i.e. so that $\|\vu-\vv\|\le\eta$. The minimal cardinality of an $\eta$-net of $X$, if finite, is denoted $|\calN_\eta|=\calN(X,\eta)$ and is called the covering number of $X$ at scale $\eta$.
\end{definition}
\begin{lemma}[Covering number of the complex sphere, {\cite[Lemma~14]{kopel2020random}}]
	\label{appendix_lem:lem5.2}
	For any $\eta>0$, the unit sphere of $\C^d$ contains an $\eta$-net $\calN_\eta$ whose cardinality satisfies
	\[|\calN_\eta|\le\left(1+\frac{2}{\eta}\right)^{2d}.\]
\end{lemma}
\begin{lemma}[Computing the spectral norm on a net, {\cite[ Lemma~5.4]{vershynin2010introduction}}]
	\label{appendix_lem:lem5.4}
	Let $M\in\C^{d\times d}$ be Hermitian, and let $\calN_\eta$ be an $\eta$-net of the unit sphere of $\C^d$ for some $\eta\in[0,1/2)$.
	Then
	\[\|M\|=\sup_{\|\vx\|=1}|\vx^*M\vx|\le(1-2\eta)^{-1}\sup_{\vx\in\calN_\eta}|\vx^*M\vx|.\]
\end{lemma}
%########################################################################
%########################################################################
%########################################################################
\begin{lemma}[Bentkus' one-sided tail bound
	{\cite[Theorem~1]{bentkus2003inequality}}]
	\label{lem:bentkus_one_sided}
	Let $X_1,\ldots,X_N$ be independent and identically
	distributed real-valued random variables satisfying
	\[
	\E X_j=0,\qquad X_j\le b\quad\text{almost surely},\qquad
	\E X_j^2=v^2,
	\]
	where $b\ge0$. Define
	\[	\sigma^2:=N\max\{b^2,v^2\}.\]
	Let $\Phi$ denote the distribution function of a standard normal random variable, and set
	$ c_0:=	\frac{1}{1-\Phi(\sqrt3)}<25.$
	Then, for every $y\ge0$,
	\[
	\Prob\left(\sum_{j=1}^N X_j\ge y\right)\le\min\left\{
	\exp\left(-\frac{y^2}{2\sigma^2}\right),\;	c_0\left[
	1-\Phi\left(\frac{y}{\sigma}\right)	\right]\right\}.
	\]
\end{lemma}
%########################################################################

%======================================
{
	\bibliographystyle{plainnat}
	\bibliography{fusion}
}
%======================================
\end{document}